\documentclass[12pt]{article}

\usepackage{amsmath}
\usepackage{amsthm,amssymb,bm}
\usepackage{graphicx}

\newtheorem{theorem}{Theorem}

\newtheorem{proposition}{Proposition}

\begin{document}
\setlength{\baselineskip}{6.5mm}

\begin{titlepage}
 \begin{normalsize}
  \begin{flushright}
UT-Komaba/16-8\\
KOBE-TH-16-07\\
EHIME-TH-101
 \end{flushright}
 \end{normalsize}
\vspace*{10mm}
 \begin{LARGE}
   \vspace{1cm}
   \begin{center}
Non-renormalization theorem in a lattice supersymmetric theory
and the cyclic Leibniz rule
   \end{center}
 \end{LARGE}
  \vspace{5mm}
 \begin{center}
    Mitsuhiro {\sc Kato}$^{*}$, 
            \hspace{3mm}
    Makoto {\sc Sakamoto}$^{\dagger}$ 
            \hspace{3mm}and\hspace{3mm}
    Hiroto {\sc So}$^{\ddagger}$ 
\\
      \vspace{4mm}
        ${}^{*}${\sl Institute of Physics, University of Tokyo, }\\
        {\sl Komaba, Meguro-ku, Tokyo 153-8902, Japan}\\
      \vspace{4mm}
        ${}^{\dagger}${\sl Department of Physics, Kobe University, } \\
        {\sl Nada-ku, Hyogo  657-8501, Japan}\\
      \vspace{4mm}
        ${}^{\ddagger}${\sl Department of Physics,  Ehime University,} \\
        {\sl Bunkyou-chou 2-5, Matsuyama 790-8577, Japan}\\
      \vspace{1cm}

  ABSTRACT\par
 \end{center}
 \begin{quote}
  \begin{normalsize}
$N=4$ supersymmetric quantum mechanical model is formulated on the lattice. Two supercharges, among four, are exactly conserved with the help of the cyclic Leibniz rule without spoiling the locality.
In use of the cohomological argument, any possible local terms of the effective action are classified into two categories which we call type-I and type-II, analogous to the D- and F-terms in the supersymmetric field theories.  
We prove non-renormalization theorem on the type-II terms which include mass and interaction terms
with keeping  a lattice constant finite, while type-I terms such as the kinetic terms have nontrivial quantum  corrections. 
\end{normalsize}
 \end{quote}

\end{titlepage}
\vfil\eject


\section{Introduction}

Supersymmetry is not only a fascinating idea for solving the gauge hierarchy problem, but also has many interesting properties worth investigating in its own light, one of which is nonrenormalization of $F$-terms~\cite{G-S-R,Seiberg} in the Wess-Zumino model.
In order to investigate such non-trivial properties by non-perturbative means, a lattice formulation has been desired for a long time~\cite{Giedt-Review,C-K-Review}. It is, however, difficult due to the lack of the Leibniz rule of finite difference operators on the lattice~\cite{KSS-1,LAT2012}.

In our previous paper, we proposed a novel approach where we use a finite difference operator and field products which satisfy {\it cyclic Leibniz rule} (CLR) instead of ordinary {\it Leibniz rule} (LR)~\cite{KSS-2,LAT2013,LAT2014}.
Both rules are coincide with each other in the continuum limit, i.e. they reduce to the Leibniz rule of differential operator.
With the CLR approach we realized a lattice quantum mechanical model in which kinetic and interaction terms are supersymmetric invariant separately, so that we succeeded to apply localization technique to obtain some exact results of the model. An important point we should emphasize here is that there indeed exist concrete sets of finite difference operator and field product which satisfy locality, translational invariance and the CLR. A systematic method for getting a general solution of the CLR is recently found by Kadoh and Ukita~\cite{Kado-Ukita}.

In the present paper, we proceed further to realize more supersymmetries by our approach.
We consider $N=4$ supersymmetric quantum mechanical model which is obtained by dimensional reduction of $N=2$ Wess-Zumino model in two dimensions. We will construct its lattice version which exactly preserve two supersymmetry among four thanks to the CLR. As far as our knowledge is concerned, this is the first example which exactly realize two independent supersymmetries in the lattice model\footnote{%
It is known that a Nicolai map leads to an exact nilpotent supercharge (see Ref.\cite{C-K-Review}). It is easy to construct a single Nicolai map but it is hard to find a second one. In fact, other approaches without the CLR have not succeeded in constructing two Nicolai maps or two exact supercharges for the model considered in the manuscript. On the other hand, the CLR is nothing but a necessary condition for two exact nilpotent supercharges to exist in the model. This is a reason why we require the CLR for the model.
(If we abandon the locality, we could realize a full supersymmetry on lattice. We have however restricted our considerations to a class of lattice models keeping the locality because non-local lattice models
might cause trouble in continuum limit.)
}. 
Two exactly realized supersymmetric charges satisfy maximal nilpotent subalgebra, and surprisingly are sufficient to lead us to the non-renormalization theorem in the lattice model.

Let us note that if we use Wilson term to avoid doubling problem, the term must also keep the same number of supersymmetries. In our approach Wilson term as well as ordinary mass term is constructed in use of the CLR, so that they keep all the necessary supersymmetries and are protected by the non-renormalization theorem from quantum corrections.

The paper is organized as follows. In the next section, we first construct  a lattice quantum mechanical model with two exact supersymmetries with the help of the CLR. There we introduce a kind of superfields which are slightly different from conventional ones, but useful for the subsequent discussions such as non-renormalization theorem. We then discuss on the one-particle irreducible effective action including quantum correction. Since our model contains all the necessary auxiliary fields, supersymmetries are realized off-shell and the fields are transformed linearly. Thus supersymmetry transformations of the field variables in the effective action have the same form as the elementary field variables in the tree action. This fact means that the difference operator itself appeared in the transformation has no quantum correction.
Since the realized symmetry is a maximal nilpotent subalgebra of the $N=4$ supersymmetry, we utilize cohomological argument for the classification of the supersymmetric invariant terms which is discussed in section 3. There are two categories which we call type-I and type-II each of which is an analogue of D- and F-term in the Wess-Zumino model. In section 4, we will see that this model is not a trivial one by showing one-loop quantum corrections for the type-I terms such as the kinetic term.
In section 5, we prove non-renormalization theorem for the type-II terms without taking continuum limit. Section 6 is devoted to the summary and discussion. Some useful definitions and formulae are given in appendices.

\section{$N=4$ supersymmetric complex quantum mechanics}
\subsection{Supersymmetric transformations and superfields}

We consider a lattice version of complex quantum mechanical model with $N=4$ supersymmetry.
Our model has four sets of complex fundamental lattice fields\footnote{We use a term ``field" for a dynamical variable of the model because it can be regarded as 0+1 dimensional field theory.} (i.e. eight real degrees of freedom):%
\begin{equation}
\chi_{\pm n},~\bar{\chi}_{\pm n},~\phi_{\pm n},~F_{\pm n} ,
\label{orgfd}
\end{equation}
\noindent
where $\phi_{\pm}, F_{\pm}~(\chi_{\pm}, \bar{\chi}_{\pm})$ are bosonic (fermionic) variables, and 
$n$ stands for a lattice site.  The  lattice constant is set to be unity for brevity. 
An inner product between  these fields is defined as $A\cdot B=\langle A,B \rangle \equiv \sum_nA_nB_n$.

\begin{figure}[htbp] 
\begin{center}
\includegraphics[scale=0.8]{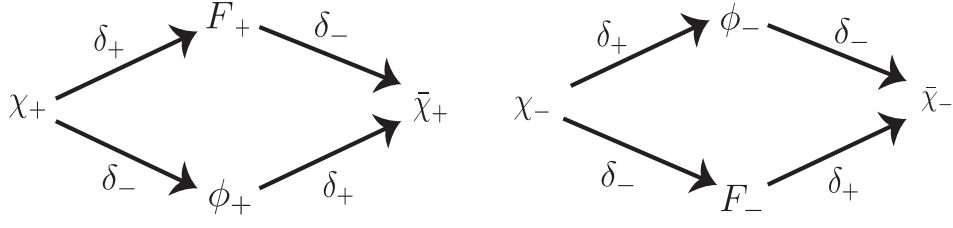}
\caption{$\delta_{\pm}$ transformation for fields}
\label{fig1-1}
\end{center}
\end{figure}         

These lattice fields are transformed under two  supersymmetries  $\delta_{\pm}$  as follows: 
\begin{eqnarray}
\left\{
 \begin{array}{l}
   \delta_{+}\phi_{+n}=\bar{\chi}_{+n}\,,\\
   \delta_{+}\bar{\chi}_{+n}=0\,,\\   
   \delta_{+}\chi_{+n}=F_{+n}\,,\\
   \delta_{+}F_{+n}=0\,,\\    
    \delta_{+}\phi_{-n}=0\,,\\
   \delta_{+}\bar{\chi}_{-n}=0\,,\\   
   \delta_{+}\chi_{-n}=-i(\nabla\phi_{-})_{n}\,,\\
   \delta_{+}F_{-n}=-i(\nabla\bar{\chi}_{-})_n\,,\\
\end{array} \right.
\qquad\qquad
\left\{
 \begin{array}{l}
   \delta_{-}\phi_{-n}=-\bar{\chi}_{-n}\,,\\
    \delta_{-}\bar{\chi}_{-n}=0 \,,\\
    \delta_{-}\chi_{-n}=F_{-n}\,,\\
    \delta_{-}F_{-n}=0 \,,\\
    \delta_{-}\phi_{+n}=0\,,\\
    \delta_{-}\bar{\chi}_{+n}=0 \,,\\
    \delta_{-}\chi_{+n}=i(\nabla\phi_{+})_n\,,\\
    \delta_{-}F_{+n}=-i(\nabla\bar{\chi}_{+})_n\,,\\   
 
 \end{array} \right.
\label{SUSYtransformation1}
\end{eqnarray}
\noindent
where $\nabla$ is a local difference operator on a field such as $(\nabla \phi)_m\equiv \sum_n\nabla_{mn}\phi_n$. 
These transformations are diagrammatically depicted 
 in  figure \ref{fig1-1}.  As we will see later, 
(\ref{SUSYtransformation1})   is exactly realized  in our  lattice theory. The remaining transformations  
in full $N=4$ SUSY are  
\begin{eqnarray}
\left\{
 \begin{array}{l}
   \bar{\delta}_{+}\phi_{+n}= 0\,,\\
   \bar{\delta}_{+}\bar{\chi}_{+n}= -i(\nabla\phi_+)_n\,,\\   
   \bar{\delta}_{+}\chi_{+n}= 0\,,\\
   \bar{\delta}_{+}F_{+n}= -i(\nabla\chi_+)_n\,,\\    
   \bar{\delta}_{+}\phi_{-n}= \chi_{-n}\,,\\  
   \bar{\delta}_{+}\bar{\chi}_{-n}= F_{-n} \,,\\   
   \bar{\delta}_{+}\chi_{-n}=0 \,,\\
   \bar{\delta}_{+}F_{-n}= 0 \,,\\
 \end{array} \right.
\qquad\qquad
\left\{
 \begin{array}{l}
   \bar{\delta}_{-}\phi_{-n}= 0\,,\\
   \bar{\delta}_{-}\bar{\chi}_{-n}= i(\nabla\phi_-)_{n} \,,\\
   \bar{\delta}_{-}\chi_{-n}= 0\,,\\
   \bar{\delta}_{-}F_{-n}=-i(\nabla\chi_-)_{n}  \,,\\
   \bar{\delta}_{-}\phi_{+n}= -\chi_{+n}\,,\\
   \bar{\delta}_{-}\bar{\chi}_{+n}= F_{+n} \,, \\
   \bar{\delta}_{-}\chi_{+n}= 0\,,\\
   \bar{\delta}_{-}F_{+n}= 0\,.\\   
   \end{array} \right.
\label{SUSYtransformation2}
\end{eqnarray}
\noindent
The SUSY transformations $\delta_{\pm},\bar{\delta}_{\pm}$ are realized by the following differential operators, 
\begin{eqnarray}
Q_+ & \equiv & \bar{\chi}_+\cdot\frac{\partial}{\partial {\phi_+}}
+{F}_+\cdot\frac{\partial}{\partial {\chi_+}}
-i\nabla \phi_-\cdot\frac{\partial}{\partial {\chi_-}}
-i\nabla \bar{\chi}_-\cdot\frac{\partial}{\partial {F_-}} ,\nonumber \\
Q_- & \equiv & i \nabla{\phi}_+\cdot\frac{\partial}{\partial {\chi_+}}
-i \nabla \bar{\chi}_+\cdot\frac{\partial}{\partial {F_+}}
- \bar{\chi}_-\cdot\frac{\partial}{\partial {\phi_-}}
+ F_-\cdot\frac{\partial}{\partial {\chi_-}},  \nonumber \\
\bar{Q}_+& \equiv & -i\nabla \phi_+\cdot\frac{\partial}{\partial {\bar{\chi}_+}}
-i\nabla \chi_+\cdot\frac{\partial}{\partial {F_+}}
+{\chi}_-\cdot\frac{\partial}{\partial {\phi_-}}
+F_-\cdot\frac{\partial}{\partial {\bar{\chi}_-}} ,\nonumber \\
\bar{Q}_-& \equiv &
-{\chi}_+\cdot\frac{\partial}{\partial {\phi_+}}
+F_+\cdot\frac{\partial}{\partial {\bar{\chi}_+}} 
+ i\nabla \phi_-\cdot\frac{\partial}{\partial {\bar{\chi}_-}}
-i\nabla \chi_-\cdot\frac{\partial}{\partial {F_-}} .
\label{SUSYtrans-2}
\end{eqnarray}
\noindent
From (\ref{SUSYtrans-2}), 
we find the following anti-commutation relations as a part of the $N=4$ SUSY algebra, 
\begin{equation}
\{Q_{\pm},\bar{Q}_{\mp}\}=\{Q_{\pm},Q_{\pm}\}=\{Q_{\pm},Q_{\mp}\}
= \{\bar{Q}_{\pm},\bar{Q}_{\pm}\}=\{\bar{Q}_{\pm},\bar{Q}_{\mp}\}=0 ,
\label{fullalgebra}
\end{equation}
which could be realized on lattice. The last piece of the algebra is given by  
\begin{equation}
\{Q_{\pm},\bar{Q}_{\pm}\}=P,\quad{\rm where}~iP \equiv \sum_{f=\phi_{\pm},F_{\pm},\chi_{\pm},\bar{\chi}_{\pm}}(\nabla f)\cdot \frac{\partial}{\partial f}. 
\label{fullalgebra2}
\end{equation}
\noindent
The problem is that $P$ cannot be realized as an exact symmetry on the lattice, because
it does not satisfy the relation $PX_n = -i(\nabla X)_n$ for general composite fields $X_n$. 
This comes from the fact that the Leibniz rule of a finite difference operator cannot be realized on the lattice 
due to the no-go theorem~\cite{KSS-1,LAT2012}.  
Therefore we must try to realize only a part of nilpotent subalgebra (\ref{fullalgebra}). A possible maximal set of supercharges is either $(Q_+,Q_-)$  or $(\bar Q_+,\bar Q_-)$.\footnote{There are other choices $(Q_+,\bar Q_-)$ and $(\bar Q_+,Q_-)$ whose algebraic structures are the same as in the main text. Their realizations, however, are different and the following arguments do not apply straightforwardly.}

For concreteness we take 
\begin{equation}
Q_+^2=Q_-^2=\{Q_+,Q_-\}=0,
\label{nilpotentSUSY}
\end{equation}
\noindent
as a maximal nilpotent subalgebra to be realized in our model. From here on we call this nilpotent-SUSY.
Our task is now to find functionals consisting of fundamental fields which satisfy
\begin{equation}
Q_{\pm}{\cal O}=0 .
\label{SUSYinv}
\end{equation}
To this end, it is helpful to examine cohomology of the nilpotent supercharges as will be seen shortly.

For later convenience, we assign a $U(1)$ charge $\pm 1$ to the fields with $\pm$ index respectively.
In addition to  this , 
we assign another $U(1)$ charge (we call it $U(1)_R$) which is defined by the eigenvalue of each field to the operator
\begin{eqnarray}
R\equiv \phi_+\cdot \frac{\partial}{\partial \phi_+} - F_+\cdot \frac{\partial}{\partial F_+}
-\phi_-\cdot \frac{\partial}{\partial \phi_-}+F_-\cdot \frac{\partial}{\partial F_-} .
\label{U(1)_R}
\end{eqnarray}
\noindent
In the holomorphy-like argument in later section, we will assign these $U(1)$ and $U(1)_R$ charges also to the parameters in the action such as coupling constants as well as mass parameter and the Wilson coefficient (a coefficient of the Wilson term). 
The quantum numbers (fermion number $N_F$, $U(1)$ and $U(1)_R$ charges)
 for fields are summarized in Table~\ref{table1}.  
 
\begin{table}[htp]
\caption{Quantum numbers for lattice fields}
\begin{center}
\begin{tabular}{|c|c|c|c|c|}\hline
 & $\phi_{\pm}$ & $\bar{\chi}_{\pm}$ & $\chi_{\pm}$ & $F_{\pm}$  \\
 \hline
 $N_F$ &0& $-1$ & 1 & 0 \\ \hline
$U(1)$ & $\pm 1$& $\pm 1$  & $ \pm 1$ &  $\pm 1$ \\  \hline
$U(1)_R$ &  $\pm 1$&$0$ &  $0$ &  $\mp 1$  \\
 \hline
\end{tabular}
\end{center}
\label{table1}
\end{table}%

We define following additional operators
\begin{eqnarray}
K_+&=&\phi_+\cdot \frac{\partial}{\partial \bar{\chi}_+}+ \chi_+\cdot \frac{\partial}{\partial F_+} ,
\nonumber \\
K_-&=&-\phi_-\cdot \frac{\partial}{\partial \bar{\chi}_-}+ \chi_-\cdot \frac{\partial}{\partial F_-} ,
\label{operator-def-K}
\end{eqnarray}
\noindent
and 
\begin{eqnarray}
N_+&\equiv&\phi_+\cdot \frac{\partial}{\partial {\phi}_+}+ \bar{\chi}_+\cdot \frac{\partial}{\partial \bar{\chi}_+}
+ \chi_+\cdot \frac{\partial}{\partial {\chi}_+}+ {F}_+\cdot \frac{\partial}{\partial {F}_+} ,
\nonumber \\
N_-&\equiv&\phi_-\cdot \frac{\partial}{\partial {\phi}_-}+ \bar{\chi}_-\cdot \frac{\partial}{\partial \bar{\chi}_-}
+ \chi_-\cdot \frac{\partial}{\partial {\chi}_-}+ {F}_-\cdot \frac{\partial}{\partial {F}_-} .
\label{operator-def-N}
\end{eqnarray}
\noindent
The $U(1)$ charge defined before is equivalent to the eigenvalue of the operator $N_+-N_-$.
These operators and supercharges $Q_{\pm}$ satisfy the following algebra
\begin{eqnarray}
&\{Q_{\pm},K_{\pm}\}=N_{\pm},~ [Q_{\pm},R]=\pm Q_{\pm},~
[K_{\pm},R]=\mp K_{\pm}, \nonumber \\
& \left[Q_{\pm},N_{\pm}\right]=[Q_{\pm},N_{\mp}]=[R,N_{\pm}]=[N_{\pm},N_{\mp}]=
[K_{\pm},N_{\pm}]=[K_{\pm},N_{\mp}]=0, \nonumber \\
&\{Q_{\pm},K_{\mp}\}=\{K_+,K_-\}=
Q_{\pm}^2=K_{\pm}^2=0, 
\label{QKRalgebra}
\end{eqnarray}
\noindent
This algebra (\ref{QKRalgebra}) is useful to analyze a cohomology of nilpotent SUSY, because $K_{\pm}$ plays an analogous role of homotopy operator.

Let us consider a monomial ${\cal O}_{\bm k}$ consists of the fundamental fields, where ${\bm k}$ stands for a collection of lattice site of each field in the monomial, e.g.~${\cal O}_{n,m,l}=\chi_{n}\phi_{m}\phi_{l}$.
And also consider a linear sum of a monomial over its lattice sites ${\cal O}=\sum_{\bm k} C_{\bm k} {\cal O}_{\bm k}$. If the coefficient $C_{\bm k}$ satisfies the following two conditions, then we say the $C_{\bm k}$ is  {\it translationally invariant local coefficient} or simply TILC.
\begin{enumerate}
\item {\it Translational invariance}:
$C_{\bm k}$ is invariant under discrete translation of lattice site, i.e. $C_{m_0+1,m_1+1,\cdots, m_I+1}=C_{m_0,m_1,\cdots, m_I}$. Then it is a function only of the site differences,
\begin{equation}
C_{m_0,m_1,\cdots,m_I} =  C(m_0-m_1,m_0-m_2,\cdots,m_0-m_{I}).
\label{coefficient}
\end{equation}
\item {\it Locality}: For each index (site difference) $k_\ell$, if $k_\ell$ is large enough, then there exist $M, L>0$ such that
\begin{equation}
|C(k_1,k_2,\cdots,k_I)| < L \exp (- M |k_{\ell}|) 
\label{locality-1}
\end{equation}
\end{enumerate}

If $C$ is TILC, then we can define holomorphic function
\begin{equation}
\tilde{C}(z_1,z_2,\cdots,z_I) \equiv  \sum_{\bm{k}}  C(k_1,k_2,\cdots,k_I)
z_1^{k_1} z_2^{k_2}\cdots z_I^{k_I}
\end{equation}
\noindent
which is holomorphic in an $I$-dimensional complex domain,
\begin{equation}
{\cal \bm{D}}^I=\{ 1-\epsilon_1 < |z_1| < 1+\epsilon_1,  1-\epsilon_2 < |z_2| < 1+\epsilon_2,\cdots ,
  1-\epsilon_I < |z_I| < 1+\epsilon_I | \epsilon_i >0\},
\end{equation}
with $\epsilon_{i} < M\ (i=1,2,\cdots,I)$.
We call $\tilde C$ holomorphic-representation or simply H-representation of $C$.

As for the fields, we introduce ``superfields" which are slightly different from conventional ones in the continuum theory, but correspond to a certain rearrangement of the original fields. With the Grassmann variables $\theta_{\pm}$, we define eight superfields $\Phi_{\pm}$, $\Psi_{\pm}$, $\Upsilon_{\pm}$ and $S_{\pm}$:
\begin{eqnarray}
\Phi_{\pm n} & \equiv & \phi_{\pm n}  \pm \theta_{\pm} \bar{\chi}_{\pm n},\qquad\Upsilon_{\pm n} \equiv F_{\pm n} -i \theta_{\mp} (\nabla \bar{\chi}_{\pm})_n, \nonumber \\ 
\Psi_{\pm n} & \equiv &\chi_{\pm n}+\theta_{\pm}\{F_{\pm n}-i\theta_{\mp}(\nabla\bar{\chi}_{\pm})_n\} \pm i \theta_{\mp} (\nabla \phi_{\pm})_n,
\qquad S_{\pm n} \equiv \bar{\chi}_{\pm n}.
\end{eqnarray}
\noindent
Then nilpotent SUSY transformations $\delta_{\pm}$ have simple expressions on the superfields, namely derivatives with respect to $\theta_{\pm}$:
\begin{equation}
\delta_{\pm}\Xi_n=\frac{\partial}{\partial \theta_{\pm}}\Xi_n,\qquad{\rm where}
\qquad \Xi_n=\Phi_{\pm n},\Psi_{\pm n},\Upsilon_{\pm n},S_{\pm n}.
\label{SFSUSY}
\end{equation}
\noindent
Figure~\ref{fig2} depicted diagrammatically these transformations of the superfields.

\begin{figure}[htbp] 
\begin{center}
\includegraphics[scale=0.8]{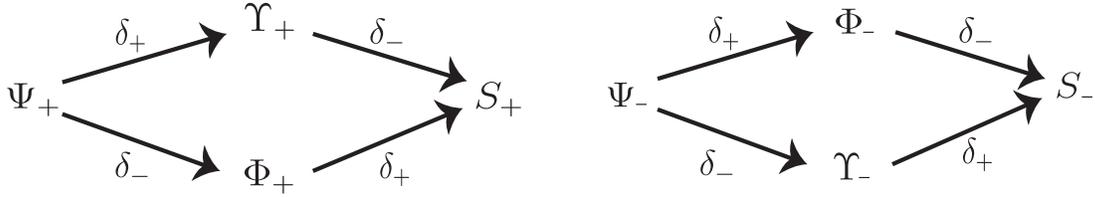}
\caption{$\delta_{\pm}$ transformation for superfields}
\label{fig2}
\end{center}
\end{figure}         

For any functional ${\cal O} $ made of the above superfields, the nilpotent supersymmetry transformations act as
\begin{equation}
Q_{\pm}{\cal O}=\frac{\partial}{\partial \theta_{\pm}}{\cal O} .
\label{SUSYSF}
\end{equation}
\noindent
Furthermore, any functional  $O$ of original fields (\ref{orgfd}) can be expressed in terms of that of superfields  
\begin{equation}
O = \int d^2\theta \theta_+\theta_- {\cal O} ={\cal O}\vert_{\theta_{\pm}=0} ,
\label{intrepSF}
\end{equation}
where we used the following integral formulae for the Grassmann variables $\theta_{\pm}$, 
\begin{equation}
\int d^2\theta \equiv \int d\theta_-d\theta_+,
\quad\int d\theta_+\theta_+=\int d\theta_-\theta_-=1,
\quad\int d^2\theta \theta_+\theta_-=1 .
\label{integralformula}
\end{equation}
\noindent
Note that our superfields are merely a change of variables from the original fields; the number of superfields is the same as that of  the original fields, and the expression of any quantity in terms of original fields can be recovered from that of superfields by putting $\theta_{\pm}=0$.
\noindent
Nonetheless, because of the simplicity of the expression (\ref{SFSUSY})  of $Q_{\pm}$,  the relations (\ref{SUSYSF}) and (\ref{intrepSF})  will largely simplify the 
classification of   supersymmetric invariant local functionals  in the fermion number zero sector in later section.
Therefore, we shall use these  eight superfields instead of eight fundamental fields 
 for the discussion of the $Q_{\pm}$-cohomology and 
for a proof of nonrenormalization theorem.  
Quantum numbers for these superfields and $\theta_{\pm}$ are assigned as listed in Table \ref{table2}.


\begin{table}[htbp]
\caption{Quantum numbers for superfields and $\theta_{\pm}$}
\begin{center}
\begin{tabular}{|c|c|c|c|c|c|}\hline
 & $\Phi_{\pm}$ & $\Psi_{\pm}$ & $\Upsilon_{\pm}$ & $S_{\pm}$  & $\theta_{\pm}$ \\
 \hline
 $N_F$ &0& 1 & 0 & $-1$ &  1 \\ \hline
$U(1)$ & $\pm 1$& $\pm 1$ & $\pm 1$ & $\pm 1$ &  0 \\  \hline
$U(1)_R$ &  $\pm 1$& $0$ & $\mp 1$ &$0$ &  $\pm 1$  \\
 \hline
\end{tabular}
\end{center}
\label{table2}
\end{table}%


\subsection{Nilpotent-SUSY invariant action and cyclic Leibniz rule}
Now, we write down a kinetic term $S_0$ of the nilpotent-SUSY invariant tree action, 
\begin{eqnarray}
S_0&=&\langle\nabla \phi_-,\nabla \phi_+\rangle+i\langle\nabla\bar{\chi}_-,\chi_+\rangle 
- i\langle{\chi}_-,\nabla\bar{\chi}_+\rangle
  - \langle F_-,F_+\rangle\nonumber \\
  &=&  \int d^2\theta \langle \Psi_-,\Psi_+\rangle  .
\label{kinetic-term}
\end{eqnarray}
Due to the translational invariance and locality, the difference operator $\nabla$ in the above is actually a function of the site difference, i.e.~$\nabla_{mn}= \nabla(m-n)$, and also is bounded for large $|k|$ as $|\nabla(k)|<L\exp(-M|k|)$ with positive real numbers $L,M$; in other words, it is TILC. 

In general, this kinetic term may have a doubling problem. In order to avoid it we introduce Wilson terms as well as ordinary mass terms,
\begin{eqnarray}
S_m&=&    \langle  F_+,G_+\phi_+   \rangle 
-      \langle  \chi_+,G_+\bar{\chi}_+   \rangle 
+ \langle  F_-,G_-\phi_-   \rangle 
+      \langle  \chi_-,G_-\bar{\chi}_-   \rangle \nonumber   \\
&=& -\int d^2 \theta \theta_- \langle \Psi_+, G_+\Phi_+ \rangle 
+ \int d^2 \theta \theta_+ \langle \Psi_-, G_-\Phi_- \rangle .
\label{mass-term}
\end{eqnarray}
\noindent
Here $G_{\pm}$ expresses mass and Wilson terms for $\pm$ fields. 
As a typical example, in the case of $(\nabla)_{mn}=(\delta_{m+1, n} -\delta_{m-1, n})/2$, we can take the term
\begin{eqnarray}
(G_+)_{mn} = m_+\delta_{m,n} + r_+(\delta_{m+1, n}-2\delta_{m,n}+\delta_{m-1, n})/2 ,\nonumber \\
(G_-)_{mn} = m_-\delta_{m,n} + r_-(\delta_{m+1, n}-2\delta_{m,n}+\delta_{m-1, n})/2, 
\label{Wilsonterm-1}
\end{eqnarray}
\noindent 
where $m_{\pm}$ are mass parameters with $m_+^*=m_-$ and $r_{\pm}$ are Wilson parameters with $r_+^*=r_-$. 
In the  H-representation, they are expressed as 
\begin{equation}
\tilde{\nabla}(w)=(w-w^{-1})/2,~\tilde{G}_{\pm}(w) = m_{\pm} + r_{\pm}(w-2 +w^{-1})/2 .
\label{Wilsonterm-2}
\end{equation}
Note that, as far as site indices are symmetric, any choice of $(G_{\pm})_{mn}=(G_{\pm})_{nm}$ makes $S_m$ invariant under nilpotent SUSY (\ref{SUSYtransformation1}) with the symmetric difference operator $(\nabla)_{mn}=-(\nabla)_{nm}$.\footnote{Though confusing, difference operator with anti-symmetric site indices is usually called symmetric.}

Let us turn to the interaction terms. We begin with the simplest case, namely three point interaction which can be written in the form
\begin{eqnarray}
S_{int}&=&  \lambda_+\langle F_+,\phi_+*\phi_+   \rangle
-  \lambda_+\langle \chi_+,\bar{\chi}_+*\phi_+   \rangle
-  \lambda_+\langle \chi_+,\phi_+*\bar{\chi}_+  \rangle \nonumber \\
 &&+  \lambda_-\langle F_-,\phi_-*\phi_-   \rangle
+  \lambda_-\langle \chi_-,\bar{\chi}_-*\phi_-   \rangle
+  \lambda_-\langle \chi_-,\phi_-*\bar{\chi}_-  \rangle \nonumber   \\
&=&  -\int d^2 \theta \theta_- \lambda_+\langle \Psi_+, \Phi_+*\Phi_+ \rangle 
+ \int d^2 \theta \theta_+ \lambda_-\langle \Psi_-, \Phi_-*\Phi_- \rangle 
\label{int-term}
\end{eqnarray}
\noindent
where we denote a two-field product as
\begin{equation}
 (A*B)_{\ell}\equiv\sum_{mn}M_{\ell mn}A_{m}B_{n}.
 \label{*product}
\end{equation}
Here the coefficient $M_{\ell mn}$ is symmetric with respect to the last two indices, $M_{\ell mn}=M_{\ell nm}$, so that two bosonic fields are commutative and two fermionic fields are anti-commutative to each other.
Three-point interaction action (\ref{int-term}) is supersymmetric invariant if the product (\ref{*product}) satisfies the next relation for any bosonic fields $A, B, C$,
\noindent
\begin{equation}
\langle (\nabla A),B*C \rangle+ \langle (\nabla B),C*A\rangle+\langle (\nabla C),A*B\rangle=0 ,
\end{equation}
\begin{equation}
{\rm i.e.}\quad
\sum_k\left(\nabla_{k\ell}M_{kmn}+\nabla_{km}M_{kn\ell}+\nabla_{kn}M_{k\ell m}\right)=0.
\end{equation}
\noindent
This is the cyclic Leibniz rule (CLR) proposed in \cite{KSS-2} and plays a crucial role in our formulation. One of the simplest choices is 
\begin{eqnarray}
M_{\ell mn} &=&\frac{1}{6}(2\delta_{\ell, m-1} \delta_{\ell, n-1} 
+ \delta_{\ell, n-1} \delta_{\ell, m+1}+ \delta_{\ell, n+1}\delta_{\ell, m-1} 
+ 2 \delta_{\ell, m+1}\delta_{\ell, n+1}), \nonumber \\
\nabla_{mn} &=& \frac{1}{2}(\delta_{m, n-1} - \delta_{m, n+1}).
\label{eg-1}
\end{eqnarray}

In order to extend to multi-point interaction more than three, we introduce  multi-field symmetric product as an extension of (\ref{*product}),
\begin{equation}
[\![B^{(1)},B^{(2)},\cdots,B^{(k)}]\!]_m=\sum_{n_1\cdots  n_{k}}M_{mn_1\cdots n_{k}}B^{(1)}_{n_1}\cdots B^{(k)}_{n_{k}}
 \label{kproduct}
\end{equation}
where the last $k$ indices of the coefficient $M$ are totally symmetric.
The above $*$-product corresponds to the $k=2$ case: $(A*B)_m=[\![A,B]\!]_m$.
Again if this multi-field product satisfies the following CLR for any $k+1$ bosonic fields $B^{(i)}$
\begin{eqnarray}
\langle\nabla B^{(0)},[\![B^{(1)},B^{(2)},\cdots,B^{(k)}]\!]\rangle&+&
\langle\nabla B^{(1)},[\![B^{(2)},B^{(3)},\cdots,B^{(0)}]\!]\rangle
\nonumber\\
+\,\cdots &+&
\langle\nabla B^{(k)},[\![B^{(0)},B^{(1)},\cdots,B^{(k-1)}]\!]\rangle=0,
\label{kCLR}
\end{eqnarray}
\begin{equation}
{\rm i.e.}\quad
\sum_{p:\,{\rm cyclic\,perms\,of\,}\{0,1,\cdots,k\}}\sum_m\nabla_{mn_{p(0)}}M_{mn_{p(1)}\cdots n_{p(k)}}=0,
\end{equation}
then a $(k+1)$-point interaction term like
\begin{equation}
-\lambda^{(k+1)}_+\int d^2\theta\,\theta_-\langle\Psi_+,[\![\Phi_+,\cdots,\Phi_+]\!]\rangle
+\lambda^{(k+1)}_-\int d^2\theta\,\theta_+\langle\Psi_-,[\![\Phi_-,\cdots,\Phi_-]\!]\rangle
\end{equation}
is invariant under nilpotent SUSY.

\subsection{Effective action on lattice and Ward-Takahashi identities}

We will discuss quantum corrections for our model in the subsequent sections. In order to make it clear of what kind of quantum symmetry our arguments treat, here we define the effective action and see its invariance under the supersymmetry.
Let us begin by deriving the Ward-Takahashi identity 
for the 1-particle-irreducible(1PI) effective action $\Gamma(\phi)$ which is defined by the Legendre transformation 
$\Gamma(\phi)=\sum_{i,n}J^i_n\phi^i_n-W(J)$ from the generating functional $W(J)$ for the connected Green's functions. 
$W(J)$ is defined by the following path integral with source fields $J^i_n$
\begin{equation}
e^{W(J)}=\int{\cal D}\varphi\, e^{-S(\varphi)+\sum_{i,n}J^i_n\varphi^i_n}.
\label{genfunc}
\end{equation}
Here we use an index $i$ for distinguishing various fields and $n$ for lattice sites. 
The argument field $\phi^i_n$ of $\Gamma(\phi)$ is defined as an expectation value of $\varphi^i_n$ under the source $J$:
\begin{equation}
\phi^i_n\equiv\frac{\partial W(J)}{\partial J^i_n}=\langle\varphi^i_n\rangle_J.
\end{equation}
We use left-derivative convention for fermionic fields. So we have \begin{equation}
J^i_n=(-1)^{|J^i||\phi^i|}\frac{\partial\Gamma(\phi)}{\partial\phi^i_n}
\end{equation}
with grassmann parity $|A|=0$ or $1$ for grassmann even or odd field $A$, respectively.

Let us make the change of integration fields from $\varphi$ to $\varphi+\delta\varphi$ in  eq.~(\ref{genfunc}). 
If we choose $\delta\varphi$ as an infinitesimal symmetry transformation and the action $S(\varphi)$ 
and the measure ${\cal D}\varphi$ are invariant under the transformation, then we obtain the identity
\begin{equation}
\int{\cal D}\varphi\, \sum_{i,n}J^i_n\delta\varphi^i_n\,
e^{-S(\varphi)+\sum_{i,n}J^i_n\varphi^i_n}=0,
\end{equation}
that is,
\begin{equation}
\sum_{i,n}J^i_n\langle \delta\varphi^i_n \rangle_J 
=\sum_{i,n}\langle \delta\varphi^i_n \rangle_J\frac{\partial\Gamma(\phi)}{\partial\phi^i_n}
=0  .
\label{wardid}
\end{equation}
Since our supersymmetry transformations $Q_{\pm}$ are linearly realized in the typical form of
\begin{equation}
\delta\varphi^i_n=\epsilon\sum_{j,m}A^{ij}_{nm}\varphi^j_m,
\end{equation}
\noindent
we find 
\begin{equation}
\langle\delta\varphi^i_n\rangle_J=
\epsilon\sum_{j,m}A^{ij}_{nm}\langle\varphi^j_m\rangle_J
=\epsilon\sum_{j,m}A^{ij}_{nm}\phi^j_m .
\end{equation}
Thus defining the transformation for $\phi^i_n$ as $\delta\phi^i_n=\epsilon\sum_{j,m}A^{ij}_{nm}\phi^j_m$,  
we obtain the Ward-Takahashi identity for the effective action
\begin{equation}
\delta\Gamma(\phi)\equiv
\sum_{i,n}\delta\phi^i_n\frac{\partial\Gamma(\phi)}{\partial\phi^i_n}=0 .
\end{equation}
This shows that the supersymmetry transformation for the effective action has the same form of a tree level action. 
This comes from the fact that we retain auxiliary fields and the transformation is kept linear. 
Also notice that the coefficients $A^{ij}_{nm}$ in the transformation for the original field $\varphi$ 
is inherited to those of $\phi$, which means the difference operator in the coefficients is the same both for $\varphi$ and $\phi$. 
In other words, we can just replace original fields by its expectation values in supersymmetry transformations $Q_{\pm}$ and the Ward-Takahashi identity becomes
\begin{equation}
Q_{\pm}\Gamma=0.
\label{WT-id}
\end{equation}
In the following sections we use the same symbol for the original field and its expectation value in the discussion of the effective action for simplicity.

\section{Cohomology of $Q_{\pm}$ and the classification of $Q_{\pm}$-invariant functionals of fields }

In this section, we are going to analyze cohomology of the nilpotent-SUSY transformations $Q_{\pm}$ and classify all $Q_{\pm}$-invariant functionals of fields with TILC.
As will be seen shortly, there are two categories which we call type-I and -II. Typical examples of the former are kinetic terms, while those of the latter are mass terms and interaction terms appeared in the tree action.

Since $Q_{\pm}$ do not change the number of `$+$'-fields nor that of `$-$'-fields, we can safely restrict our discussion in the subsector with definite number of `$\pm$'-fields without loss of generality. If we write ${\cal O}(n_+,n_-)$ for the functional consisting of $n_+$ `$+$'-fields and $n_-$ `$-$'-fields, then
\begin{equation}
N_{\pm}{\cal O}(n_+,n_-)=n_{\pm}{\cal O}(n_+,n_-),
\end{equation}
for the operator defined in (\ref{operator-def-N}).

Immediate consequence at this point is the following. In use of the algebra $\{Q_{\pm},K_{\pm}\}=N_{\pm}$ in (\ref{QKRalgebra}), 
$Q_{\pm}$-closed functional ${\cal O}(n_+,n_-)$ can be written in the form
\begin{eqnarray}
{\cal O}(n_+,n_-) &=& Q_+K_+{\cal O}(n_+,n_-)/n_+\qquad{\rm if }\quad Q_+{\cal O}=0~{\rm and}~n_+\ne0,\label{delta+}\\
{\cal O}(n_+,n_-) &=& Q_-K_-{\cal O}(n_+,n_-)/n_-\qquad{\rm if }\quad Q_-{\cal O}=0~{\rm and}~n_-\ne0\label{delta-}.
\end{eqnarray}
Thus $Q_+$ and $Q_-$ cohomologies in the space of functionals $\{{\cal O}(n_+,n_-)\}$ are trivial for $n_+\ne 0$ and $n_-\ne 0$ respectively.

Our next task is to determine $Q_-$ cohomology in the subspace $\{{\cal O}(n_+,0)\}$ and $Q_+$ cohomology in the subspace $\{{\cal O}(0,n_-)\}$. Before stating the result, let us define the {\it CLR terms} in the $N_F=1$ sector as
\begin{eqnarray}
{\cal C}_+^{N_F=1}&=&
\sum_{m,n_1,\cdots,n_k}C_{mn_1\cdots n_k}\chi_{+\,m}\phi_{+\,n_1}\cdots\phi_{+\,n_k}
\quad{\rm in}\quad\{{\cal O}(k+1,0)\},\nonumber\\
{\cal C}_-^{N_F=1}&=&
\sum_{m,n_1,\cdots,n_k}C_{mn_1\cdots n_k}\chi_{-\,m}\phi_{-\,n_1}\cdots\phi_{-\,n_k}
\quad{\rm in}\quad\{{\cal O}(0,k+1)\}\label{CLRterm}
\end{eqnarray}
where the coefficients $C$ are TILC and satisfy the CLR relation (\ref{kCLR}).
Note that CLR terms with $N_F=1$ have the maximum $U(1)_R$ charge (eigenvalue of the operator $R$) in the space $\{{\cal O}(k+1,0)\}$ and the minimum in $\{{\cal O}(0,k+1)\}$.

Then we have the following theorem for $N_F=1$ sector:
\noindent
\begin{theorem}  {\rm(Fundamental theorem on the cohomology of nilpotent SUSY)} \\
If ${\cal P}_{+}$  is a local functional in $\{{\cal O}(n_+,0)|n_+>0\}$ with $N_F=1$ and satisfies
\begin{equation}
Q_{-} {\cal P}_{+} =0, 
\end{equation}
\noindent
then ${\cal P}_+$ can be written in the form
\begin{equation}
{\cal P}_{+}= {\cal C}_+^{N_F=1} + Q_{-}{\cal Q}_{+},
\end{equation}
where ${\cal Q}_+$ is a local functional in $\{{\cal O}(n_+,0)|n_+>0\}$ with  $N_F=2$.\\
Similarly, 
if ${\cal P}_{-}$ is a local functional in $\{{\cal O}(0,n_-)|n_->0\}$  with  $N_F=1$ and satisfies
\begin{equation}
Q_{+} {\cal P}_{-} =0, 
\end{equation}
\noindent
then ${\cal P}_-$ can be written in the form
\begin{equation}
{\cal P}_{-}= {\cal C}_-^{N_F=1} + Q_{+}{\cal Q}_{-}.
\end{equation}
where   ${\cal Q}_{-}$ is a local functional in $\{{\cal O}(0,n_-)|n_->0\}$ with $N_F=2$.
\label{fthm}
\end{theorem}

The proof for this theorem is shown in Appendix~C with some preparations in Appendix~B.

We turn to discuss a nilpotent-SUSY invariant local functional $S$
with $N_F=0$. 
The quantity shall be useful in discussing an  effective action  and 
proving a nonrenormalization theorem. 

\begin{proposition}
 A nilpotent-SUSY invariant local functional $S$
with $N_F=0$ can be generally written as 
\begin{eqnarray}
S= Q_{+}Q_{-} T(+,-)&+&Q_{+}(+{\rm type~ CLR~terms~with~}N_F=1)\nonumber \\
 &+&Q_- (-{\rm type~ CLR~terms~with~}N_F=1), 
 \label{S0}
\end{eqnarray}
\noindent
where $T(+,-)$ is  a local functional with $N_F=2$.  
\label{prop1}
\end{proposition}

\begin{proof}
Since nilpotent SUSY-transformations $Q_\pm$ do not change the number of $+$-fields nor $-$-fields, we can separately argue each term with definite numbers of $\pm$-fields.
For a term with $n_+\ne0$ and $n_- \ne 0$ (let us denote ${\cal T}$), 
from  (\ref{delta+}) and (\ref{delta-}) 
any  nilpotent-SUSY invariant local functional   
 takes the following form,  
\begin{equation}
{\cal T}=Q_-Q_+\Big(K_+K_- {\cal T}/(n_+n_-) \Big).
\label{Tterm}
\end{equation}
\noindent
Note that both $K_+$ and $K_-$ map  from local functionals to  local functionals. 

For a term with $n_-=0$, $n_+\ne 0$ and $N_F=0$ (let us denote ${\cal U}_+$), 
any nilpotent-SUSY invariant functional 
can be written as 
\begin{equation}
{\cal U}_+ =Q_+K_+ {\cal U}_+  /n_+ .
\end{equation}
\noindent
 From   $Q_-{\cal U}_+=0$  and   $\{Q_-,K_+\}=0$ in (\ref{QKRalgebra}), 
it follows that $ {\cal P}_+ \equiv K_+ {\cal U}_+ /n_+$ is a $Q_-$-invariant local functional with $N_F=1$. 
Thus
we can apply Theorem~\ref{fthm}
to ${\cal P}_+$, so that we have 
\begin{equation}
{\cal P}_{+}=(+{\rm type~  CLR~ terms~with~}N_F=1) + Q_{-}{\cal Q}_{+} ,
\end{equation}
where ${\cal Q}_+$ is a local functional with $n_+>0$, $n_-=0$ and  $N_F=2$.
Thus
\begin{equation}
 {\cal U}_+= Q_+ {\cal P}_{+} = Q_+(+{\rm type~  CLR~ terms~with~}N_F=1) + Q_+Q_{-}{\cal Q}_{+} ,
\end{equation}
\noindent
From similar discussions,  for a nilpotent-SUSY invariant local functional ${\cal V}_-$ with  $n_+=0$, $n_-\ne 0$ and $N_F=0$, we can get 
\begin{equation}
 {\cal V}_-= Q_-(-{\rm type~  CLR~ terms~with~}N_F=1) + Q_-Q_{+}{\cal Q}_{-} ,
\end{equation}
where ${\cal Q}_-$ is a local functional with $n_+=0$, $n_->0$ and  $N_F=2$.  

Thus, putting all terms together, we have 
\begin{eqnarray}
S&=&Q_-Q_+\Big(K_+K_- {\cal T}/(n_+n_-) \Big)+
 Q_+(+{\rm type~  CLR~ terms~with~}N_F=1) + Q_+Q_{-}{\cal Q}_{+} \nonumber \\
&& +Q_-(-{\rm type~  CLR~ terms~with~}N_F=1) + Q_-Q_{+}{\cal Q}_{-} \nonumber \\
  &=&
  Q_+Q_-\Big(-K_+K_- {\cal T}/(n_+n_-) +{\cal Q}_+  - {\cal Q}_-\Big) \nonumber \\
&&  +
    Q_+(+{\rm type~  CLR~ terms~with~}N_F=1)  +
       Q_-(-{\rm type~  CLR~ terms~with~}N_F=1).   \nonumber \\
       &&   
\label{S}
\end{eqnarray}
This is indeed a form (\ref{S0}).\footnote{Note that any local functional  $O$ with $n_+\ne 0$ and $n_-\ne 0$  can be always expressed as $K_+K_- O'$ by another local functional $O'$ 
up to $Q_{\pm}$ exact form functionals, since  $1=[Q_+Q_-,K_+K_-]/(n_+n_-)+
(Q_+K_+/n_+ +  Q_-K_- /n_-)$ for $n_+\ne 0$ and $n_- \ne 0$.  }   
\end{proof}

The result (\ref{S0}) of Proposition~\ref{prop1} means that nilpotent-SUSY invariant local functionals with $N_F=0$ are 
classified into (i) $Q_+$ and $Q_-$ exact forms (type-I),  (ii) $Q_+$ exact but not $Q_-$ exact forms (type-II$_+$), 
and (iii) $Q_-$ exact but not $Q_+$ exact forms (type-II$_-$). 
The important thing is that  functionals of type-II are only CLR terms.

Let us consider a superfield counterpart of Proposition~\ref{prop1}.
From   (\ref{SUSYSF}) and (\ref{intrepSF}), we can translate a  nilpotent SUSY transformed functional 
$Q_{\pm}O$  into a functional of superfields, 
\begin{eqnarray}
Q_{\pm}O &=& \int d^2\theta \theta_+\theta_- Q_{\pm} {\cal O}  \nonumber \\
&= &\int d^2\theta \theta_+\theta_- \frac{\partial}{\partial \theta_{\pm}} {\cal O}  \nonumber \\
&=& \mp\int d^2\theta \theta_{\mp}  {\cal O} .
\label{intrepSF1}
\end{eqnarray}
\noindent
And also, 
\begin{eqnarray}
Q_{+}Q_{-} O &=&-\int d^2\theta \theta_+\theta_- 
\frac{\partial}{\partial\theta_{-}}\frac{\partial}{\partial\theta_{+}} {\cal O}  \nonumber \\
&= &  -\int d^2\theta {\cal O}.  
\label{intrepSF2}
\end{eqnarray}
By utilizing  (\ref{intrepSF1}), (\ref{intrepSF2}) and (\ref{S0}), 
the next proposition follows:

\begin{proposition}
A nilpotent-SUSY invariant local functional with $N_F=0$ of 
superfields can be written in the form  
\begin{eqnarray}
S=  \int d^2 \theta {\cal T}(+,-)- \int d^2 \theta \theta_- (+{\rm type~ CLR~terms~with~}N_F=1)\nonumber \\
 +\int d^2 \theta  \theta_+(-{\rm type~ CLR~terms~with~}N_F=1), 
 \label{S01}
\end{eqnarray}
\noindent
where the ``$\pm{\rm type~ CLR~terms~with~}N_F=1$"  consisting of  $k+1$ superfields $(k=0,1,2,\cdots)$ are given by  
\begin{equation}
\sum_{mn_1\cdots n_k} C_{mn_1\cdots n_k} \Psi_{\pm m} \Phi_{\pm n_1} \cdots  \Phi_{\pm n_k} ,
\label{CLRtermSF}
\end{equation}
\noindent
whose coefficient $C$ is a TILC  satisfying  the CLR.  
\label{prop2}
\end{proposition}
 

In Section~5, we will use a trick in which we lift constant parameters in the model, such as coupling constants, mass parameters and Wilson-term coefficients, to constant superfields by introducing constant super-partner for each constant parameter. Therefore we need an extended version of Proposition~\ref{prop2} including constant superfields.
Let us take $N_{CF}$ constant parameters $\rho^i_{\pm}$ ($i=1,\cdots,N_{CF}$) as constant fields, and their super-partners $\zeta^i_{\rho \pm}$ ($i=1,\cdots, N_{CF}$) as well.
Then $Q_{\pm},K_{\pm},N_{\pm}$ are modified to include constant fields:
\begin{equation}
Q_{\pm}'\equiv Q_{\pm} \pm \sum_{i}\zeta^i_{\rho \pm}\frac{\partial}{\partial \rho^i_{\pm}},~
K_{\pm}'\equiv K_{\pm} \pm \sum_{i}\rho^i_{\pm}\frac{\partial}{\partial \zeta^i_{\rho \pm}},~
N_{\pm}'\equiv N_{\pm} + \sum_{i}(\rho^i_{\pm}\frac{\partial}{\partial \rho^i_{\pm}}+ 
\zeta^i_{\rho \pm} \frac{\partial}{\partial \zeta^i_{\rho \pm}}),
\label{extension}
\end{equation}
\noindent
and their algebra is the same as before: 
\begin{eqnarray}
&\{Q'_{\pm},K'_{\pm}\}=N'_{\pm},\nonumber \\
& \left[Q'_{\pm},N'_{\pm}\right]=[Q'_{\pm},N'_{\mp}]=[N'_{\pm},N'_{\mp}]=
[K'_{\pm},N'_{\pm}]=[K'_{\pm},N'_{\mp}]=0, \nonumber \\
&\{Q'_{\pm},K'_{\mp}\}=\{K'_+,K'_-\}=
Q'^2_{\pm}=K'^2_{\pm}=0.
\label{QKRalgebra2}
\end{eqnarray}

We assign $N_F=0$ for $\rho^i_{\pm}$ and $-1$ for $\zeta^i_{\rho \pm}$.  
Then, $\pm$type local functionals are defined by new operators as  
\begin{equation}
N_{\mp}' {\cal O}_{\pm}=0, \qquad N_{\pm}' {\cal O}_{\pm} \ne 0 .
\end{equation}
\noindent
We should notice that the constant fields do not depend on lattice site, so the coefficient appeared in functionals only depend on the site index of ordinary fields to which the notion of locality refers.

From (\ref{extension}), we can easily show an extended version of the fundamental theorem as follows,
\begin{theorem}
Let ${\cal P}'_{\pm}$  be  any $\pm$type local functional with $N_F=1$.
If 
\begin{equation}
Q'_{\mp} {\cal P}'_{\pm} =0, 
\end{equation}
\noindent
then 
\begin{equation}
{\cal P}'_{\pm}= (\pm{\rm type~  CLR~ extended~ terms~with~}N_F=1) + Q'_{\mp}{\cal Q}'_{\pm} ,
\end{equation}
\noindent
where    $\pm$type CLR extended terms  with $N_F=1$ of $k+1$-th order are expressed as 
\begin{equation}
\sum_{mn_1\cdots n_k} C'_{mn_1\cdots n_k} \chi_{\pm m} \phi_{\pm n_1} \cdots  \phi_{\pm n_k} ,
\label{CLRterm2}
\end{equation}
\noindent
and the coefficient $C'$ is a TILC  satisfying the CLR condition and may depend on  
constant fields $\rho_{\pm}$.  ${\cal Q}'_{\pm}$ are $N_F=2$ local functionals which may depend on
$\rho^i_{\pm},\zeta^i_{\rho \pm}$. 
\end{theorem}

\noindent
The proof goes almost same way as that of the previous Theorem~1
with a care of the fact that $\rho^i_{\pm},\zeta^i_{\rho \pm}$ are invariant under $\delta_{\mp}$. 
We also note that $\zeta^i_{\rho \pm}$ does not contribute 
to nontrivial CLR terms with $N_F=1$,  like  $\bar{\chi}_{\pm}$. 

From this theorem, we can prove two additional propositions:

\begin{proposition}
 A nilpotent-SUSY invariant local functional $S'$
with $N_F=0$ is generally written in the form
\begin{eqnarray}
S'= Q'_{+}Q'_{-} T'(+,-)&+&Q'_+({\rm +type~ CLR~extended~ terms~with~}N_F=1)\nonumber \\
 &+&Q'_- ({\rm -type~ CLR~extended~ terms~with~}N_F=1), 
 \label{S03}
\end{eqnarray}
\noindent
where $T'(+,-)$ is  a local functional with $N_F=2$.
\label{prop3}
\end{proposition}

By introducing two kinds of constant superfields $\rho^i_{\pm}(\theta_{\pm}) 
\equiv \rho^i_{\pm} \pm \theta_{\pm}\zeta^i_{\rho \pm}$ and $\zeta^i_{\rho \pm}(\theta_{\pm})\equiv\zeta^i_{\rho \pm}$, 
the action of $Q_{\pm}'$ on any superfield functional ${\cal O}'$ is expressed simply as
\begin{equation}
Q_{\pm}' {\cal O}' = \frac{\partial}{\partial \theta_{\pm}} {\cal O}' .
\label{newSUSYSF}
\end{equation}
\noindent
Similarly to (\ref{intrepSF}),  this   expression (\ref{newSUSYSF}) enables us 
to replace component fields in Proposition 3 to  
superfields.

\begin{proposition}
A nilpotent-SUSY invariant $N_F=0$ local functional of 
superfields   can be written as  
\begin{eqnarray}
S'=  \int d^2 \theta {\cal T}'(+,-)
&-& \int d^2 \theta \theta_- (+{\rm type~ CLR~extended~terms~with~}N_F=1)\nonumber \\
 &+&\int d^2 \theta  \theta_+(-{\rm type~ CLR~extended~terms~with~}N_F=1), 
 \label{S04}
\end{eqnarray}
\noindent
where the above $\pm{\rm type~ CLR~extended~ terms~with~}N_F=1$   of $k+1$-th order is defined as 
\begin{equation}
\sum_{mn_1\cdots n_k} C'_{mn_1\cdots n_k}(\rho^i_{\pm}\big(\theta_{\pm})\big) \Psi_{\pm m} \Phi_{\pm n_1} \cdots  \Phi_{\pm n_k} ,
\label{CLRtermSF4}
\end{equation}
\noindent
and the coefficient $C'$ is a TILC  satisfying the CLR condition and may depend on  constant superfields 
$\rho^i_{\pm}(\theta_{\pm})$.
\label{prop4}
\end{proposition}

In summary of this section, all nilpotent-SUSY invariant local functionals with $N_F=0$ are classified   
into type-I and  type-II. The latter contains only 
a linear combination of CLR terms (\ref{CLRterm}), (\ref{CLRtermSF}), (\ref{CLRterm2}) 
and (\ref{CLRtermSF4}) as a consequence of  $Q_{\mp}$-cohomology. 
It is a surprising result that there is no nilpotent-SUSY invariant  type-II
local functional with $N_F=1$ 
except CLR terms.
For other sectors, we can find nontrivial cohomology elements in $N_F<1$ sector due to the existence of $S_{\pm}$ and $\Phi_{\pm}$, while in $N_F>1$ sector there is no nilpotent-SUSY invariant  type-II local functional.

\section{Quantum effects for the model}
Before discussing non-renormalization theorem in the subsequent section, let us look at the quantum correction explicitly, say, in one-loop level.
The kinetic term, which is type-I functional, indeed get contributions from one-loop diagrams. On the other hand, the terms given by the type-II functionals such as mass terms and the interaction terms have no contributions from one-loop diagrams. Even for two-loop level there is no quantum correction to the type-II terms due to the CLR. 

Let us recapitulate the action for the kinetic term $S_0$ and the mass term $S_m$ defined in subsection~2.2:
\begin{eqnarray}
\!\!\!S_0&=&\langle\nabla \phi_-,\nabla \phi_+\rangle+i\langle\nabla\bar{\chi}_-,\chi_+\rangle 
- i\langle{\chi}_-,\nabla\bar{\chi}_+\rangle
  - \langle F_-,F_+\rangle,   \nonumber \\
  &=&  \int d^2\theta \langle \Psi_-,\Psi_+\rangle  ,
\label{kinetic-term2}
\end{eqnarray}
\noindent 
\begin{eqnarray}
\,\,S_m &=&   \langle F_+,G_+\phi_+\rangle - \langle\chi_+,G_+\bar{\chi}_+\rangle 
 + \langle F_-,G_-\phi_-\rangle+\langle\chi_-,G_-\bar{\chi}_-\rangle  \nonumber \\
 &=& -\int d^2 \theta \theta_- \langle \Psi_+, G_+\Phi_+ \rangle 
+ \int d^2 \theta \theta_+ \langle \Psi_-, G_-\Phi_- \rangle ,
\end{eqnarray}
\noindent
where $G_{\pm}$ include Wilson terms (\ref{Wilsonterm-1}),(\ref{Wilsonterm-2}). 

We could think the mass term as one of the interaction terms and might pursuit our perturbative calculation with massless propagator. In order to avoid infrared divergences, instead, we use perturbation with massive propagator. 
We denote a symbol $\langle\cdots\rangle_0 $ as an expectation value  in the tree-level. 
By defining $D_{mn}=D_{nm}\equiv(\nabla^T\nabla + G_-G_+)^{-1}_{mn}$, the tree-level propagators can be written as
\begin{eqnarray}
\langle \phi_{\mp m} \phi_{\pm n}\rangle_0 &=& 
D_{mn},\qquad\quad\langle\phi_{\pm m}\phi_{\pm n}\rangle_0 = 0 ,\nonumber \\
\langle\chi_{\mp m} \bar{\chi}_{\pm n}\rangle_0&=& -i (\nabla D)_{mn},\quad
\langle\chi_{\pm m}\bar{\chi}_{\pm n}\rangle_0=  
\pm(G_{\mp}D)_{mn}, \nonumber \\
\langle F_{\mp m} F_{\pm n} \rangle_0 &=& - (\nabla D \nabla^T)_{mn},~\langle F_{\pm m} \phi_{\pm n} \rangle_0
=  (G_{\mp}D)_{mn},
\label{2point}
\end{eqnarray}
where $(\nabla^T)_{nm}=(\nabla)_{mn}$.
\noindent
Note that our mass operators $G_+,G_-=(G_+)^{\dagger}$ include Wilson terms, thus have site-dependence in general. 
In terms of superfields, relevant 2-point functions are
\begin{eqnarray}
\langle \Phi_{\mp m}(\theta)\Phi_{\pm n}(\theta')  \rangle_0 &=&D_{mn},~
\langle \Phi_{\pm m}(\theta)\Phi_{\pm n}(\theta')  \rangle_0 =0,~\nonumber \\
\langle \Psi_{\mp m}(\theta) \Phi_{\pm n}(\theta') \rangle_0 &=& 
\mp i\delta(\theta_{\pm} -\theta'_{\pm})(\nabla D)_{mn}, \nonumber \\
 \langle \Psi_{\pm m}(\theta) \Phi_{\pm n}(\theta') \rangle_0 &= &
 \delta(\theta_{\pm} - \theta'_{\pm}) (G_{\mp}D)_{mn},
  \nonumber \\
\langle \Psi_{\mp m}(\theta) \Psi_{\pm n}(\theta') \rangle_0 &=&  
i\delta^2(\theta-\theta')(\nabla D\nabla^T)_{mn}, \nonumber \\
 \langle \Psi_{\pm m} (\theta)\Psi_{\pm n}(\theta') \rangle_0 &=&  
i \delta^2(\theta-\theta')(\nabla DG^T_{\mp})_{mn}
=-i \delta^2(\theta-\theta')(G_{\mp} D\nabla^T)_{mn},
\label{2pointSF}
\end{eqnarray}
\noindent
where $\delta$-functions for the Grassmann variables are defined as   $\delta (\theta_{\pm}) \equiv \theta_{\pm},\delta^2(\theta) \equiv \theta_+\theta_-$.

For the interaction term, we only consider three-point interaction for simplicity:
\begin{eqnarray}
S_{int}  &=&  \lambda_+(\langle F_+,\phi_+*\phi_+\rangle-2 \langle\chi_+,\bar{\chi}_+*\phi_+\rangle) \nonumber\\
 && +\lambda_-(\langle F_-,\phi_-*\phi_-\rangle+2 \langle\chi_-,\bar{\chi}_-*\phi_-\rangle) \nonumber \\
 & = & -\lambda_+ \int d^2\theta \theta_- \langle \Psi_+,\Phi_+ * \Phi_+\rangle
 + \lambda_- \int d^2\theta \theta_+ \langle \Psi_-,\Phi_- * \Phi_-\rangle, 
\end{eqnarray}
\noindent
where $\lambda_{\mp}=\lambda_{\pm}^*$ and the coefficient $M$ in the definition of $*$ product (\ref{*product}) is a TILC satisfying the CLR condition,
\begin{eqnarray}
\sum_k(\nabla_{k \ell} M_{kmn}  + \nabla_{km} M_{kn\ell}+ \nabla_{kn} M_{k\ell m} )
=\sum_k(\nabla_{\ell k}^T M_{k mn}  + \nabla_{mk}^T M_{kn\ell}+ \nabla_{nk}^T M_{k\ell m}) =0. && \nonumber \\
&&
\label{CLR2}
\end{eqnarray}

\subsection{Corrections for kinetic terms}

One-loop corrections to the kinetic term $F_-F_+$, for instance, are given by Figure~\ref{fig1}.
\begin{figure}[htbp] 
\begin{center}
\includegraphics[scale=.5]{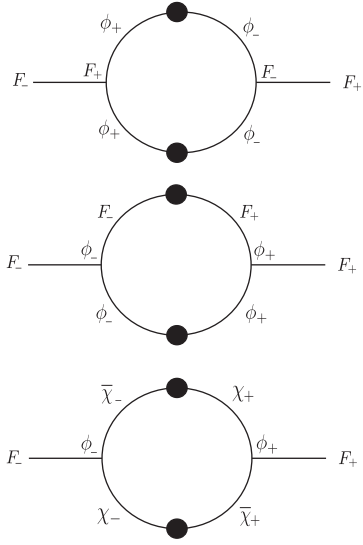}
\caption{One-loop corrections  for $F_- {F}_+$. }
\label{fig1}
\end{center}
\end{figure}         
Here internal lines with filled circles symbolize massive propagators.

 At the one-loop level, we  can write down the $F_-F_+$ propagator with correction in the form
\begin{eqnarray}
\langle F_{-m} F_{+n} \rangle = &&  \langle F_{-m} F_{+n} \rangle_0 
+  \langle F_{-m} F_{+n'} \rangle_0  \Sigma_{FF}^{+-}(n',m') \langle F_{-m'} F_{+n} \rangle_0 \nonumber \\
&& + \langle F_{-m} \phi_{-n'} \rangle_0
 \Sigma_{\phi\phi}^{-+}(n',m')
 \langle \phi_{+m'} F_{+n} \rangle_0  
\end{eqnarray}
\noindent
and the self-energies at the one-loop are explicitly given by 
\begin{eqnarray}
\Sigma_{FF}^{+-}(n,m) &=& 2\lambda_+\lambda_-M_{nk\ell}M_{mij}D_{ik}D_{j\ell} ,\label{FFself}
\end{eqnarray}
\noindent
and 
\begin{eqnarray}
 \Sigma_{\phi\phi}^{-+}(n,m) &= &-4\lambda_+\lambda_-M_{\ell kn}M_{ijm}(\nabla D \nabla^T)_{\ell i} D_{kj} \nonumber \\
 && - 4\lambda_+\lambda_-M_{k\ell n}M_{ijm}(\nabla D)_{kj} (D\nabla^T)_{\ell i} \nonumber \\
  &=&4\lambda_+\lambda_-(M\nabla)_{nk\ell }(\nabla^TM)_{ijm}D_{\ell i} D_{kj} \nonumber \\
  &=& -2  \lambda_+\lambda_-(M\nabla)_{nk\ell }(\nabla^TM)_{mij}D_{j\ell} D_{ik} \nonumber \\
  &=&     - (\nabla^T \Sigma_{FF}^{+-}\nabla)(n,m)             ,
\end{eqnarray}
\noindent
where we have used the symmetric property of $D_{mn}=D_{nm}$ in (\ref{2point}) and  the CLR 
relation  (\ref{CLR2}).  Short-hand notations $\nabla_{\ell k}M_{\ell mn}\equiv(\nabla^TM)_{kmn}$ and 
$M_{k\ell m}\nabla_{mn}=M_{km\ell}\nabla_{mn}\equiv(M\nabla)_{k\ell n}$ are also used. 

\begin{figure}[htbp] 
\begin{center}
\includegraphics[scale=.5]{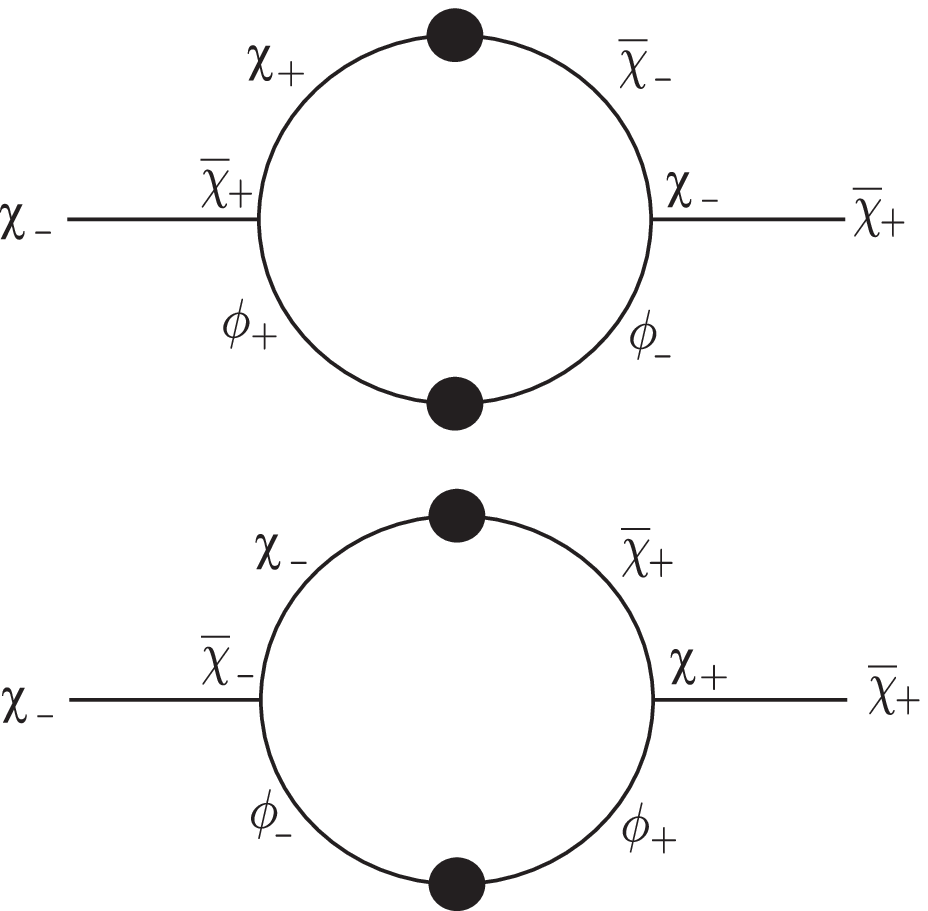}
\caption{One-loop corrections  for $\chi_- \bar{\chi}_+$. }
\label{fig2-1}
\end{center}
\end{figure}         
The $\chi_-\bar{\chi}_+$ propagator is  evaluated  at the one-loop level as 
\begin{eqnarray}
\langle \chi_{-m} \bar{\chi}_{+n} \rangle &=&   \langle {\chi}_{-m} \bar{\chi}_{+n} \rangle_0 
+  \langle \chi_{-m} \bar{\chi}_{+n'} \rangle_0
 \Sigma_{\bar{\chi}\chi}^{+-}(n',m')
 \langle \chi_{-m'} \bar{\chi}_{+n} \rangle_0 \nonumber \\
  && + \langle \chi_{-m} \bar{\chi}_{-n'} \rangle_0
 \Sigma_{\bar{\chi}\chi}^{-+}(n',m')
 \langle \chi_{+m'} \bar{\chi}_{+n} \rangle_0 ,
\end{eqnarray}
\noindent
where
\begin{eqnarray}
\Sigma_{\bar{\chi}\chi}^{+-}(n,m) &=& 4i\lambda_+\lambda_-M_{kn\ell}M_{mij}(D\nabla)_{ik}D_{j\ell} \nonumber \\
 & = & -2i\lambda_+\lambda_-(M\nabla)_{nk\ell}M_{mij}D_{ik}D_{j\ell} \nonumber \\
 & = & -i(\nabla^T\Sigma_{FF}^{+-})(n,m) .
\end{eqnarray}
\noindent
In this calculation the  CLR (\ref{CLR2}) is important. 
Similarly, 
\begin{eqnarray}
 \Sigma_{\bar{\chi}\chi}^{-+}(n,m) & = & 4i\lambda_+\lambda_-M_{\ell kn}M_{mij}(\nabla D)_{\ell j} D_{ki} \nonumber \\
& = & -i(\nabla^T\Sigma_{FF}^{+-})(n,m) = \Sigma_{\bar{\chi}\chi}^{+-}(n,m) .
\end{eqnarray}

\begin{figure}[htbp] 
\begin{center}
\includegraphics[scale=.5]{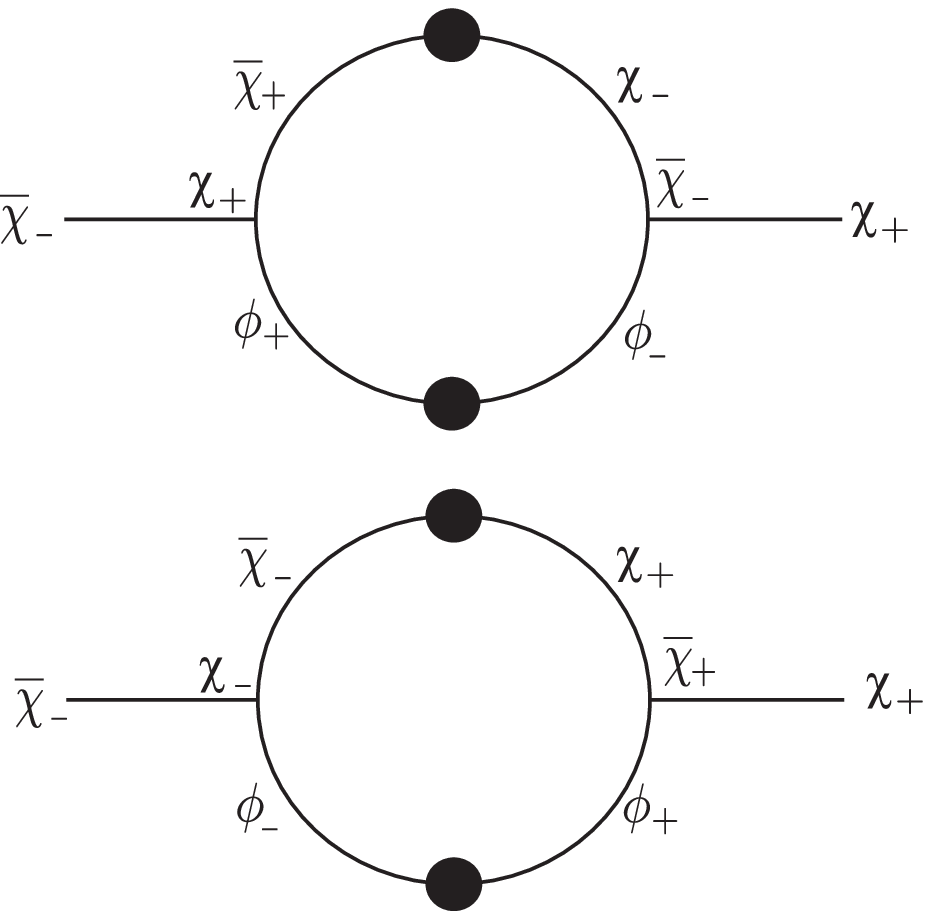}
\caption{One-loop corrections  for $\bar{\chi}_- {\chi}_+$. }
\label{fig3-1}
\end{center}
\end{figure}         
For $\bar{\chi}_-{\chi}_+  $   at the one-loop level, the propagator is evaluated as 
\begin{eqnarray}
\langle \bar{\chi}_{-m} {\chi}_{+n} \rangle &=&   \langle \bar{\chi}_{-m} {\chi}_{+n} \rangle_0 
+  \langle \bar{\chi}_{-m} {\chi}_{+n'} \rangle_0
 \Sigma_{{\chi}\bar{\chi}}^{+-}(n',m')
 \langle \bar{\chi}_{-m'} {\chi}_{+n} \rangle_0 \nonumber \\
  && + \langle \bar{\chi}_{-m} {\chi}_{-n'} \rangle_0
 \Sigma_{{\chi}\bar{\chi}}^{-+}(n',m')
 \langle \bar{\chi}_{+m'} {\chi}_{+n} \rangle_0 ,
\end{eqnarray}
\noindent
where 
\begin{equation}
\Sigma_{{\chi}\bar{\chi}}^{+-}(n,m) = 4i\lambda_+\lambda_-M_{nk\ell}M_{imj}(\nabla D)_{ik}D_{j\ell} ,
\label{chichi-1}
\end{equation}
\noindent
and 
\begin{equation}
 \Sigma_{{\chi}\bar{\chi}}^{-+}(n,m) = 4i\lambda_+\lambda_-M_{nk\ell}M_{jim}(\nabla D)_{j\ell} D_{ki} .
\label{chichi-2}
\end{equation}
\noindent
In use of the CLR  (\ref{CLR2}), self-energy parts (\ref{FFself}), (\ref{chichi-1}) and (\ref{chichi-2}) are related as 
\begin{equation}
\Sigma_{\chi\bar{\chi}}^{-+}(n,m) = \Sigma_{\chi\bar{\chi}}^{+-}(n,m)= -i (\Sigma_{FF}^{+-}\nabla)(n,m). 
\end{equation}

\begin{figure}[htbp] 
\begin{center}
\includegraphics[scale=.5]{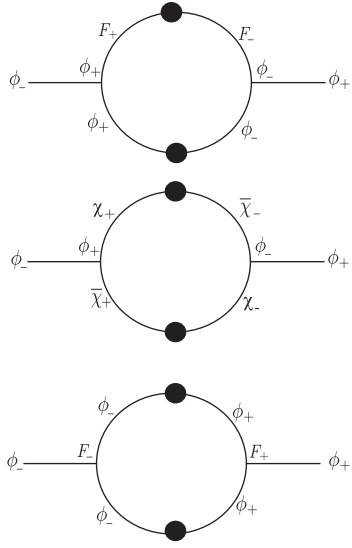}
\caption{One-loop corrections  for ${\phi}_- {\phi}_+$. }
\label{fig4}
\end{center}
\end{figure}         
Finally, the $\phi_-\phi_+$ propagator  at the one-loop level becomes 
\begin{eqnarray}
\langle \phi_{-m} \phi_{+n} \rangle &=&   \langle \phi_{-m} \phi_{+n} \rangle_0 +  \langle \phi_{-m} \phi_{+n'} \rangle_0
 \Sigma_{\phi\phi}^{+-}(n',m')
 \langle \phi_{-m'} \phi_{+n} \rangle_0 \nonumber \\
 && + \langle \phi_{-m} F_{-n'} \rangle_0
 \Sigma_{FF}^{-+}(n',m')
 \langle F_{+m'} \phi_{+n} \rangle_0 . 
\end{eqnarray}
\noindent

In summary,  we have evaluated one-loop contributions to the part of self-energy which mixes $\pm$-fields and thus to the type-I functionals, i.e.~kinetic terms.
Relevant diagrams are shown in Figures ~\ref{fig1}, 
\ref{fig2-1}, \ref{fig3-1} and \ref{fig4}.
\noindent 
Note that
\begin{equation}
\Sigma^{+-}_{FF}(n,m)=\Sigma^{+-}_{FF}(m,n),~\Sigma^{+-}_{\phi\phi}(n,m)=\Sigma^{+-}_{\phi\phi}(m,n) .
\end{equation}
\noindent
Also we can see each self-energy is related to each other 
by using  the CLR  (\ref{CLR2})  repeatedly, 
\begin{equation}
(\nabla^T\Sigma^{+-}_{FF})(n,m)=(\nabla^T\Sigma^{-+}_{FF})(n,m)
 =i (\Sigma^{\pm\mp}_{\bar{\chi}\chi})(n,m),
 \end{equation}
\begin{equation}
(\Sigma^{+-}_{FF}\nabla)(n,m)=(\Sigma^{-+}_{FF}\nabla)(n,m)
 =i (\Sigma^{\pm\mp}_{{\chi}\bar{\chi}})(n,m) ,
\end{equation}
\begin{equation}
\Sigma^{+-}_{\phi\phi}(n,m)=\Sigma^{-+}_{\phi\phi}(n,m)
 =-i (\Sigma^{+-}_{\bar{\chi}\chi}\nabla)(n,m) ,
\end{equation}
\begin{equation}
\Sigma^{+-}_{\phi\phi}(n,m)=
\Sigma^{-+}_{\phi\phi}(n,m)
 =-i (\Sigma^{+-}_{{\chi}\bar{\chi}}\nabla^T)(m,n) .
\end{equation}
\noindent
These equations  are exact relations thanks to the CLR and are nothing but a part of the  SUSY Ward-Takahashi identities.

\subsection{Quantum effects on type-II functionals}

Let us turn to the type-II functionals, the simplest of which is the mass term. So we have to evaluate self-energy parts which do not mix $\pm$-fields.
Immediate consequences we obtain are the vanishing of $\Sigma^{++}_{F\phi}$ and $\Sigma^{++}_{\chi\bar\chi}$ due to the vanishing of $\langle\phi_+\phi_+\rangle_0$ from (\ref{2point}). (We can easily see from Figure~\ref{fig5} that these self-energy parts are proportional to $\langle\phi_+\phi_+\rangle_0$.)
\begin{figure}[htbp] 
\begin{center}
\includegraphics[scale=.5]{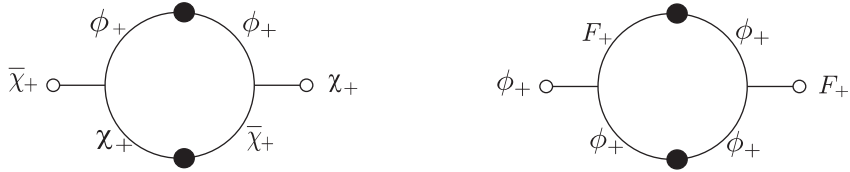}
\caption{Diagrams contributing to $\Sigma^{++}_{\chi\bar\chi}$ and
$\Sigma^{++}_{F\phi}$ where open circle stands for amputated propagator. Both diagrams are proportional to $\langle\phi_+\phi_+\rangle_0$.}
\label{fig5}
\end{center}
\end{figure}         

We emphasize here that not only pure mass term but also Wilson term have no quantum correction at one-loop (and actually any loop order as will be shown in next section), thus the Wilson parameter is unchanged.

The situation becomes a bit nontrivial in two-loop order.
\begin{figure}[htbp] 
\begin{center}\hspace{-30mm}
\includegraphics[scale=.5]{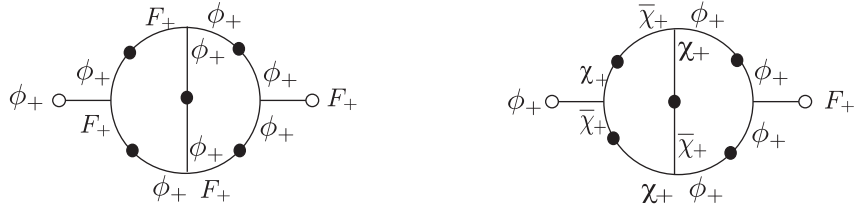}
\caption{$O(\lambda_+^4)$ contribution to $\Sigma^{++}_{F\phi}$ at the 2-loop level.  
These graphs are proportional to $(\langle \phi_+\phi_+\rangle_0)^2$.}
\label{fig7}
\end{center}
\end{figure}         
\begin{figure}[htbp] 
\begin{center}\hspace{-30mm}
\includegraphics[scale=.5]{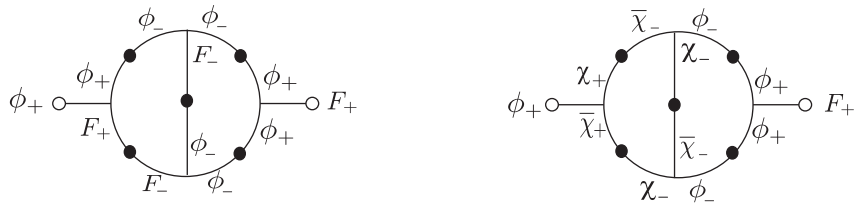}
\caption{$O(\lambda_+^2\lambda_-^2)$ contribution to $\Sigma^{++}_{F\phi}$ at 2-loop level.}
\label{fig8}
\end{center}
\end{figure}         
Actually, for the diagrams in Figure~\ref{fig7},  the correction for $F_+ \phi_+$ is proportional to $\lambda_+^4\langle  \phi_+ \phi_+ \rangle^2_0$ and 
thus it vanishes as before. On the other hand, diagrams in Figure~\ref{fig8} with $\lambda^2_+\lambda_-^2$ can be shown to be
proportional to the difference operator $\nabla$  by the CLR (\ref{CLR2}),  
\begin{eqnarray}
\Sigma^{++}_{F\phi}(n,m)|_{\mbox{\scriptsize 2-loop}}&=& -2^4  (\lambda_+\lambda_-)^2    \Big((\nabla^T M)_{inj} + (\nabla^T M)_{jni}\Big)
 M_{mk\ell}   \bar{M}_{abc}(\bar{M}\nabla)_{def}          \nonumber \\
&  & \times  D_{bj}D_{ck}D_{f\ell}D_{di}D^{-F\phi}_{ae}  \nonumber \\ 
& =&  2^4  (\lambda_+\lambda_-)^2  (\nabla^T M)_{nij}  M_{mk\ell}   \bar{M}_{abc}(\bar{M}\nabla)_{def} 
D_{bj}D_{ck}D_{f\ell}D_{di}D^{-F\phi}_{ae} ,\nonumber \\
&&
\end{eqnarray}
\noindent
where $  D_{ij}^{\pm F\phi} \equiv \langle  F_{\pm i} \phi_{\pm j}  \rangle_0$. 
Since $n$ is the site index  of external $\phi_+$, the correction to the effective action is proportional to $\nabla\phi_+$.  
Therefore, this effect can be shown to be  a 2-loop contribution to type-I functional.

\section{Non-renormalization  theorem on the lattice}
\subsection{Perturbative proof of the nonrenormalization theorem for type-II functionals}

%
%
%

%
In the previous section, we explicitly saw that there is no quantum correction in mass and Wilson terms at the one loop level, 
while the kinetic terms have nontrivial quantum corrections.  
In this section, we show the non-renormalization theorem
that the type-II terms  are not suffered from quantum
corrections at all.

Our starting action in the tree level is
%
\begin{eqnarray}
S = S_{\footnotesize \textrm{type-I}} + S_{\footnotesize \textrm{type-II}}\,,
\label{nonrenormalization_01}
\end{eqnarray}
%
where
%
\begin{eqnarray}
\hspace{-8mm}
S_{\footnotesize \textrm{type-I}}
 &=&  \int d\theta_{-}d\theta_{+} \langle \Psi_{-}\,,\,\Psi_{+}\rangle\,,
     \label{nonrenormalization_02} \\
\hspace{-8mm}
S_{\footnotesize \textrm{type-II}}
 &=& \int d\theta_{-}d\theta_{+}
     \Big\{ \theta_{-}\langle\,\Psi_{+}\,,\,
      W_{+}^{\footnotesize \textrm{tree}}(\Phi_{+}; v_{+}, m_{+}, \lambda_{+})\,\rangle
      \notag\\
 && \hspace{23mm}     
      -\, \theta_{+}\langle\,\Psi_{-}\,,\,
      W_{-}^{\footnotesize \textrm{tree}}(\Phi_{-}; v_{-}, m_{-}, \lambda_{-})\,\rangle 
     \Big\}
     \label{nonrenormalization_03}
\end{eqnarray}
with
%
\begin{eqnarray}
W_{\pm}^{\footnotesize \textrm{tree}}(\Phi_{\pm}; v_{\pm}, m_{\pm}, \lambda_{\pm})
 = -v_{\pm} - m_{\pm} G_{\pm}\Phi_{\pm} - \lambda_{\pm} \Phi_{\pm}*\Phi_{\pm}\,.
     \label{nonrenormalization_04}
\end{eqnarray}
%
Here we have a little bit generalized from the preceding sections by adding a linear term $v_{\pm}\Psi_{\pm}$ in the action as a type-II functional. The inclusion of the linear term is not a necessary thing but the argument below goes well either with or without  this term.
Also we have slightly changed our notation as $G_{\pm} \rightarrow m_{\pm}G_{\pm}$ from (\ref{Wilsonterm-1}), and so this is accompanied by the change of the Wilson coefficient $r_{\pm} \rightarrow r_{\pm}/m_{\pm}$.

Type-II part of the effective action with full quantum corrections can be separated into the sum of tree action $S_{\footnotesize \textrm{type-II}}$ and quantum correction $\Delta\Gamma^{\footnotesize \textrm{eff}}_{\footnotesize \textrm{type-II}}$
%
%
\begin{eqnarray}
\Gamma^{\footnotesize \textrm{eff}}_{\footnotesize \textrm{type-II}}
  = S_{\footnotesize \textrm{type-II}}
    + \Delta\Gamma^{\footnotesize \textrm{eff}}_{\footnotesize \textrm{type-II}}\,.
\label{nonrenormalization_05}
\end{eqnarray}
%

First of all, we are going to show that the latter should be of the form
%
\begin{eqnarray}
\hspace{-3mm}
\Delta\Gamma^{\footnotesize \textrm{eff}}_{\footnotesize \textrm{type-II}}
  &=& \int d\theta_{-}d\theta_{+}
     \Big\{ \theta_{-}\langle\,\Psi_{+}\,,\,
        W_{+}^{\footnotesize \textrm{eff}}(\Phi_{+}; v_{+}, m_{+}, \lambda_{+})\,\rangle
        \notag\\
  && \hspace{16.3mm}
           - \,\theta_{+}\langle\,\Psi_{-}\,,\,
        W_{-}^{\footnotesize \textrm{eff}}(\Phi_{-}; v_{-}, m_{-}, \lambda_{-})\,\rangle
             \Big\}. \hspace{4mm}
\label{nonrenormalization_06}
\end{eqnarray}
%
Here emphasis is on that $W_{+}^{\footnotesize \textrm{eff}}$
($W_{-}^{\footnotesize \textrm{eff}}$) depends only on
$\Phi_{+}, v_{+}, m_{+}, \lambda_{+}$ ($\Phi_{-}, v_{-}, m_{-}, \lambda_{-}$) 
but not on $\Phi_{-}, v_{-}, m_{-}, \lambda_{-}$ ($\Phi_{+}, v_{+}, m_{+}, \lambda_{+}$), and
this holomorphic property (even in the parameter dependence) will turn out to be crucial in the proof of the non-renormalization theorem. 

To this end, we replace the parameters $v_{\pm}$, $m_{\pm}$ and $\lambda_{\pm}$ 
by \textit{constant superfields}
such a way as \cite{Seiberg}%
\footnote{%
In Ref.\cite{Seiberg}, the mass parameters and the coupling constants have
been replaced by chiral superfields having the space-time coordinate-dependence.
Here, we assume that $v_{\pm}(\theta_{\pm})$, $m_{\pm}(\theta_{\pm})$ 
and $\lambda_{\pm}(\theta_{\pm})$
are the functions of the Grassmann coordinates $\theta_{\pm}$ but do not
have the site-dependence of the lattice.
This is because the SUSY invariance on the lattice will be lost 
if they depend on the site.
}
%
\begin{eqnarray}
v_{\pm}\ &\longrightarrow&\ 
  v_{\pm}(\theta_{\pm}) = v_{\pm} + \theta_{\pm}\zeta^{v}_{\pm} , \nonumber\\
m_{\pm}\ &\longrightarrow&\ 
  m_{\pm}(\theta_{\pm}) = m_{\pm} + \theta_{\pm}\zeta^{m}_{\pm} , \nonumber\\
\lambda_{\pm}\ &\longrightarrow&\ 
  \lambda_{\pm}(\theta_{\pm}) = \lambda_{\pm} + \theta_{\pm}\zeta^{\lambda}_{\pm}\, ,  
\label{nonrenormalization_07}
\end{eqnarray}
%
where the lowest components of $v_{\pm}(\theta_{\pm})$, $m_{\pm}(\theta_{\pm})$ and
$\lambda_{\pm}(\theta_{\pm})$ correspond to the original parameters
$v_{\pm}$, $m_{\pm}$ and $\lambda_{\pm}$.
Note that the tree action (\ref{nonrenormalization_01}) is still
supersymmetric under the replacement (\ref{nonrenormalization_07})
with the supersymmetry transformations:
%
\begin{eqnarray}
Q_{\pm} v_{\pm}(\theta_{\pm}) 
  &=& \frac{\partial}{\partial\theta_{\pm}} v_{\pm}(\theta_{\pm})\,, 
      \qquad \delta_{\mp} v_{\pm}(\theta_{\pm}) = 0\,,\nonumber\\
Q_{\pm} m_{\pm}(\theta_{\pm}) 
  &=& \frac{\partial}{\partial\theta_{\pm}} m_{\pm}(\theta_{\pm})\,, 
      \qquad \delta_{\mp} m_{\pm}(\theta_{\pm}) = 0\,,\nonumber\\
Q_{\pm} \lambda_{\pm}(\theta_{\pm}) 
  &=& \frac{\partial}{\partial\theta_{\pm}} \lambda_{\pm}(\theta_{\pm})\,,
      \qquad\  \delta_{\mp} \lambda_{\pm}(\theta_{\pm}) = 0\,.
\label{nonrenormalization_08}
\end{eqnarray}
%
Then, it is straightforward application of Proposition~\ref{prop4} to show that 
$\Delta\Gamma^{\footnotesize \textrm{eff}}_{\footnotesize \textrm{type-II}}$
 can be supersymmetric only if
$W_{+}^{\footnotesize \textrm{eff}}$ ($W_{-}^{\footnotesize \textrm{eff}}$)
is independent of 
$\Phi_{-}(\theta_{-}), v_{-}(\theta_{-}), m_{-}(\theta_{-}), \lambda_{-}(\theta_{-})$
($\Phi_{+}(\theta_{+}), v_{+}(\theta_{+}), m_{+}(\theta_{+})$, 
$\lambda_{+}(\theta_{+})$).

As a next step, we further 
restrict the form of the functions
$W_{\pm}^{\footnotesize \textrm{eff}}(\Phi_{\pm}; v_{\pm}, m_{\pm}, \lambda_{\pm})$
by assigning the fermion number $N_{F}$, the $U(1)$ and $U(1)_{R}$
charges to the fields, as listed in Table \ref{tab1.1}.
%
%
%
%
\begin{table}[!t]
\centering
\begin{tabular}[h]{c||c|c|c|c|c|c|c|c|c|c|c|c}
	\parbox[c][0.7cm][c]{0cm}{}%
	     & $\Psi_{\pm}$ & $\Phi_{\pm}$ & $v_{\pm}$ & $m_{\pm}$ & $\lambda_{\pm}$ & $\theta_{\pm}$ 
	 & $d\theta_{\pm}$ 
	 & $S$
	 & $\Delta\Gamma^{\footnotesize \textrm{eff}}_{\footnotesize \textrm{type\,II}}$
	 & $\frac{v_{\pm}\lambda_{\pm}}{(m_{\pm})^{2}}$
	 & $\frac{\lambda_{\pm}\Phi_{\pm}}{m_{\pm}}$
	 & $\frac{v_{\pm}}{m_{\pm}\Phi_{\pm}}$ \\   
	   \hline\hline
	\parbox[c][0.7cm][c]{0cm}{}%
	$N_{F}$ & $+1$ & $0$ & 0 & 0 & 0 & $+1$ & $-1$ & 0 & 0 & 0 & 0 & 0 \\ 
	   \hline
	\parbox[c][0.7cm][c]{0cm}{}%
	$U(1)$ & $\pm 1$ & $\pm 1$ & $\mp 1$ & $\mp 2$ & $\mp 3$ & 0 & 0 & 0 & 0 & 0 & 0 & 0\\ \hline
	\parbox[c][0.7cm][c]{0cm}{}%
	$U(1)_{\footnotesize \textrm{R}}$ 
	& 0 & $\pm 1$ & $\pm 1$ & 0 & $\mp 1$ & $\pm 1$ & $\mp 1$ & 0 & 0 & 0 & 0 & 0
\end{tabular}
\caption{%
the fermion number $N_{F}$, $U(1)$ and $U(1)_{\footnotesize \textrm{R}}$
charges
}
\label{tab1.1}
\end{table}
%
%
%
Then, 
$\Delta\Gamma^{\footnotesize \textrm{eff}}_{\footnotesize \textrm{type-II}}$
is found to be generally expressed as
%
\begin{eqnarray}
&&\Delta\Gamma^{\footnotesize \textrm{eff}}_{\footnotesize \textrm{type-II}}
      [\Psi_{\pm}, \Phi_{\pm}; v_{\pm}, m_{\pm}, \lambda_{\pm}] 
=\int d\theta_{-}d\theta_{+}
     \bigg\{ \theta_{-} 
             \langle\,\Psi_{+}\,,\,{\textstyle \frac{(m_{+})^{2}}{\lambda_{+}}}
             f_{+}\big({\textstyle  \frac{v_{+}\lambda_{+}}{(m_{+})^{2}}, 
                   \frac{\lambda_{+}\Phi_{+}}{m_{+}}}\big)\,\rangle 
                   \notag\\
&&\hspace{75mm}
           - \theta_{+} 
             \langle\,\Psi_{-}\,,\,{\textstyle \frac{(m_{-})^{2}}{\lambda_{-}}}
             f_{-}\big({\textstyle  \frac{v_{-}\lambda_{-}}{(m_{-})^{2}}, 
                   \frac{\lambda_{-}\Phi_{-}}{m_{-}}}\big)\,\rangle
             \bigg\},\hspace{6mm}
\label{nonrenormalization_09}
\end{eqnarray}
%
where $f_{\pm}(z_{\pm}, w_{\pm})$ are some holomorphic functions of the complex variables
$z_{\pm} = \frac{v_{\pm}\lambda_{\pm}}{(m_{\pm})^{2}}$ and
$w_{\pm} = \frac{\lambda_{\pm}\Phi_{\pm}}{m_{\pm}}$. 

\vspace{5mm}
In this perturbative calculation, we may expand the functions $f_{\pm}(z_{\pm}, w_{\pm})$ 
in powers of $z_{\pm}$ and $w_{\pm}$ as
%
\begin{eqnarray}
&&\Delta\Gamma^{\footnotesize \textrm{eff}}_{\footnotesize \textrm{type-II}}
      [\Psi_{\pm}, \Phi_{\pm}; v_{\pm}, m_{\pm}, \lambda_{\pm}] 
      \nonumber\\
&&\hspace{5mm} =\int d\theta_{-}d\theta_{+} \sum_{k}\sum_{l}
     \bigg\{ \theta_{-} \frac{(m_{+})^{2}}{\lambda_{+}}
             \,a_{kl}^{+}
             \bigg(\frac{v_{+}\lambda_{+}}{(m_{+})^{2}}\bigg)^{k}
             \bigg(\frac{\lambda_{+}}{m_{+}}\bigg)^{l}
             \langle\,\Psi_{+}\,,\,
             [\![\,\underbrace{\Phi_{+},\Phi_{+}, \cdots , \Phi_{+}}_{l}\,]\!]\,\rangle
             \nonumber\\
&&        \hspace{25mm}
            -\, \theta_{+} \frac{(m_{-})^{2}}{\lambda_{-}}
             \,a_{kl}^{-}
             \bigg(\frac{v_{-}\lambda_{-}}{(m_{-})^{2}}\bigg)^{k}
             \bigg(\frac{\lambda_{-}}{m_{-}}\bigg)^{l}
             \langle\,\Psi_{-}\,,\,
             [\![\,\underbrace{\Phi_{-},\Phi_{-}, \cdots , \Phi_{-}}_l\,]\!]\,\rangle
             \bigg\},\hspace{6mm}
\label{nonrenormalization_12}
\end{eqnarray}
%
\noindent
where $a^{\pm}_{kl}$ are some constant coefficients.  
The SUSY invariance then requires the CLR relations
%
\begin{eqnarray}
\langle\,\nabla\Phi_{\pm}\,,\, [\![\Phi_{\pm},\Phi_{\pm}, \cdots , \Phi_{\pm}]\!]\,\rangle
  = 0\,,
\label{nonrenormalization_13}
\end{eqnarray}
%
as they should be.
Since we are considering perturbation theory, in which the weak coupling limits of 
$v_{\pm}, \lambda_{\pm} \to 0$ are assumed to exist,
the powers of $k$ and $l$ should be restricted to 
%
\begin{eqnarray}
k \ge 0\,,\quad
l \ge 0\,,\quad
k+l \ge 1\,,
\label{nonrenormalization_14}
\end{eqnarray}
%
in order for 
$\Delta\Gamma^{\footnotesize \textrm{eff}}_{\footnotesize \textrm{type-II}}$
to be non-singular.
Note that the functions $f_{\pm}$ in (\ref{nonrenormalization_09}) 
are non-singular in the weak
coupling limits of $v_{\pm}, \lambda_{\pm} \to 0$ with (\ref{nonrenormalization_14}).

Let us next clarify what kind of diagrams lead to the terms given in 
(\ref{nonrenormalization_12}).
To this end, we use the topological relation for Feynman diagrams
%
\begin{eqnarray}
L = I - V_{1} - V_{3} + 1\,,
\label{nonrenormalization_15}
\end{eqnarray}
%
where $L, I, V_{1}$ and $V_{3}$ denote the numbers of loops, internal lines,
$v_{\pm}$-vertices and $\lambda_{\pm}$-vertices of a Feynman diagram, respectively.
Further, we have the relation
%
\begin{eqnarray}
V_{1} + 3V_{3} = E + 2I,
\label{nonrenormalization_16}
\end{eqnarray}
%
because we have one-point vertices $v_{\pm}$ and three-point vertices $\lambda_{\pm}$
in the tree action.
Here, $E$ denotes the number of external lines.
From (\ref{nonrenormalization_15}) and (\ref{nonrenormalization_16}),
we obtain
%
\begin{eqnarray}
L = - \frac{V_{1}}{2} + \frac{V_{3}}{2} - \frac{E}{2} + 1\,.
\label{nonrenormalization_17}
\end{eqnarray}
%
Since the values of $V_{1}$, $V_{3}$ and $E$ for each term in (\ref{nonrenormalization_12}) 
are given by $V_{1}=k$, $V_{3}=k+l-1$ and $E=l+1$,
corresponding diagrams turn out to be tree ones without loops, i.e.
%
\begin{eqnarray}
L = -\frac{k}{2} + \frac{k+l-1}{2} - \frac{l+1}{2} + 1 = 0\,.
\label{nonrenormalization_18}
\end{eqnarray}
%
Therefore, all the terms given in (\ref{nonrenormalization_12}) have to be 
excluded from 
$\Delta\Gamma^{\footnotesize \textrm{eff}}_{\footnotesize \textrm{type-II}}$
because the effective potential 
$\Delta\Gamma^{\footnotesize \textrm{eff}}_{\footnotesize \textrm{type-II}}$
consists of 1PI diagrams only.
This implies that there is no quantum correction to the tree type-II action 
$S_{\footnotesize \textrm{type-II}}$, i.e.
%
\begin{eqnarray}
\Gamma^{\textrm{\footnotesize eff}}_{\footnotesize \textrm{type-II}}
  = S_{\footnotesize \textrm{type-II}}\,.
\label{nonrenormalization_19}
\end{eqnarray}
%
This completes the perturbative proof of the nonrenormalization theorem on the lattice.

\vspace{3mm}

In the essential point of our proof, any nilpotent-SUSY invariant local functional
is  classified into type-I or type-II.  
In the addition,  possible  type-II local functionals are only 
terms so called as CLR-type.\footnote{
For massive perturbative calculations, major relations (\ref{nonrenormalization_15}), 
(\ref{nonrenormalization_16}),  (\ref{nonrenormalization_17}) and (\ref{nonrenormalization_18}) are unchanged.  
We simply replace  $f_{\pm}(w,z)$ with $f_{\pm}(w,z, m_+m_-)$ in (\ref{nonrenormalization_09}). 
As the result, we can obtain the same nonrenormalization theorem.}

\subsection{Consideration of nonperturbative  nonrenormalization property for type-II functional}
We briefly consider a non-perturbative justification beyond a perturbative proof. 
Let us assume that no massless mode appears even non-perturbatively in the massive theory. 
The assumption leads us that our theory is  sigularity-free at the origin in the coupling constant space with weak fields.
In our quantum mechanical model, the assumption seems natural.  But for higher dimensional cases, it we may need more careful treatment.

Any way, from the assumption,  the complex analysis in the previous subsection tells us that
the holomorphic functions $f_{\pm}(z_{\pm},w_{\pm})$  in (\ref{nonrenormalization_09}) 
equals to zero in a neighborhood  of the origin,
since the perturbative proof in the previous subsection holds there.
Then $f_{\pm}(z_{\pm},w_{\pm})$  vanish identically on the \textit{whole} complex plane because of the
analytical continuation or  the identity theorem  in  complex analysis. 
This suggests our nonrenormalization theorem holds nonperturbatively. 

%
\section{Summary and discussion}
In this article, we have constructed  a supersymmetric complex quantum mechanics model on lattice. The action is invariant under 
two nilpotent-SUSY transformations ($Q_{\pm}$) which form a maximal nilpotent subalgebra of full $N=4$ SUSY and they keep a holomorphic property.  
Furthermore, they enable us to study $Q_{\pm}$ cohomology exactly. 
 As the result of the analysis, we can classify  all local $Q_{\pm}$-invariant functionals  with $N_F=0$ into 
 type-I (such as kinetic terms) and type-II (such as a mass term including the Wilson term and interaction terms.) 
 The local functionals in  a nilpotent-SUSY invariant effective action  is also classified into the two types and 
 we have proved nonrenormalization theorem for the type-II local functionals at {\it any order of perturbative expansion} and
 {\it without taking continuum limit, namely with finite lattice constant}.  
This means that 
 we are able to realize more than one nilpotent-SUSY and the holomorphy   
 even in a regularized theory and this is extremely nontrivial result. 
  For nonperturbative justification of the nonrenormalization property of type-II functionals, we gave reasonable arguments.

The reasons why we succeeded to prove the theorem are (1) 
familiarity between  two nilpotent-SUSY transformations and holomorphy, (2) definition of local functionals, (3) existence of the CLR,
(4) cohomological analysis of the nilpotent SUSY. 
It should be noted that 
 our model is quite similar to  the $N=1,~D=4$ Wess-Zumino model which has a  $F$-term nonrenormalization theorem, 
 kinetic terms suffering with quantum corrections. 
There remain several issues worth investigation such as higher-dimensional extension, cohomological analysis of 
 local functionals with any fermion number and so on.

%
%
Finally, we would like to emphasize notable features of our lattice models.  
Note that supersymmetry in continuum theories is not the Lagrangian symmetry 
but the action one. 
That is, supersymmetric Lagrangians are invariant under 
supersymmetry transformations up to total divergences, in general. 
Our lattice model holds this property. Under the supersymmetry transformations 
(\ref{SUSYtransformation1}), 
our Lagrangian on lattice is not invariant but the action becomes invariant 
with the CLR by taking the summation over the lattice sites. 
Thus, our lattice model mimics an important feature of 
continuum supersymmetric theories. 
We would like to notice that other lattice supersymmetric models without 
the CLR do not possess this property and fail in inheriting a property 
with respect to translations in supersymmetry.

Another notable feature of our lattice model is as follows. 
Our lattice model has been shown to be cohomologically non-trivial. 
This property turns out to be crucial to prove the non-renormalization theorem 
for F-terms, as discussed in the manuscript. 
All cohomologically trivial terms can be written into exact forms and 
hence those terms are trivially invariant under nilpotent supersymmetry 
transformations without the summation over the lattice sites. 
On the other hand, every cohomologically non-trivial term cannot be written 
as any exact form and does vanish under nilpotent supersymmetry ones only after 
the summation over the lattice sites with the CLR. 
Thus the CLR plays an essential role in the non-triviality of cohomology 
for our lattice model. 
This has not been shared by other approaches and hence is an advantage of 
our lattice model with the CLR.
%
%

%
%
%
\section*{Acknowledgements}
This work is supported in part by the Grant-in-Aid for Scientific 
Research No.25287049 (M.K.), No.15K05055 (M.S.) and No.25400260 (H.S.)
by the Japanese Ministry of Education, Science, Sports and Culture.



%

\appendix

\section{Notations, H-representation for locality and some formulae}

A lattice space coordinate is expressed by an integer as 
$n, $ where 
$-N_L < n \le N_L$ and 
$2N_L$ is the lattice size.  The lattice constant $a$ is set to unity. 
A translationally-invariant  and local coefficient(TILC)  
is conveniently  expressed by a holomorphic function, H-representation, for example   
\begin{equation}
\tilde{A}(z_1,z_2,z_3)\equiv \sum_{k\ell m} A_{k \ell m n}z_1^{k-n}z_2^{\ell -n}z_3^{m-n} .
\end{equation}
The locality implies that $\tilde{A}(z_1,z_2,z_3)$ is holomorphic in a domain 
$
{\cal D} =\{1-\epsilon <|z_i| < 1+ \epsilon\, |\, \epsilon >0,~i=1,2,3  \} 
$.  Although we refer \cite{KSS-1} and \cite{KSS-2} for the detailed arguments on locality, 
we note that the meaning of locality include not only ultralocality but also exponential 
damping. 
In this article, a terminology, {\it local functional} or {\it functional with locality}, used as  a collection  of   
fields  with a local coefficient.  
A local difference operator $\nabla_{mn}$ in this paper 
is also a TILC 
\begin{equation}
\nabla_{mn}=\nabla(m-n),~\tilde{\nabla}(z)\equiv \sum_m z^m \nabla(m),
\end{equation}
\noindent
and $\tilde{\nabla}(z)$ is a holomorphic function in a domain $1-\epsilon <|z| < 1+\epsilon$ for small  $\epsilon >0$ 
with the property 
\begin{equation}
\tilde{\nabla}(z=1)=0
\end{equation}
\noindent
which corresponds to its vanishing property for  constant functions. 

The symbols $()$ and $\{\}$ for indices stand for  symmetrization and anti-symmetrization, respectively. 
 A hat symbol $\hat{}$ above an index in a sequence of indices means omission of 
 the index. For examples, $A_{(ab)}=\frac{1}{2}(A_{ab}+A_{ba}),~A_{\{ab\}}=\frac{1}{2}(A_{ab}-A_{ba}), \hat{a}bc=bc$ and $ab\hat{c}d=abd$.

The hermitian conjugations of fields are defined as 
\begin{eqnarray}
\left\{
\begin{array}{l}
\phi_{\pm}^{\dagger}=\phi_{\mp} \,, \\
F_{\pm}^{\dagger}=F_{\mp}  \,, \\
\chi_{\pm}^{\dagger}=\chi_{\mp} \,, \\
\bar{\chi}_{\pm}^{\dagger}=\bar{\chi}_{\mp} \,, \\
\end{array}\right.
\qquad\qquad
\left\{
\begin{array}{l}
\Phi_{\pm}^{\dagger}=\Phi_{\mp} \,, \\
\Upsilon_{\pm}^{\dagger}=\Upsilon_{\mp} \,, \\
\Psi_{\pm}^{\dagger}=\Psi_{\mp} \,, \\
S_{\pm}^{\dagger}=S_{\mp} \,, \\
\theta_{\pm}^{\dagger}=\theta_{\mp} \,. \\
\end{array} \right.
\end{eqnarray}
\noindent
From this hermiticity, we can show the reality of our action with (\ref{kinetic-term}), (\ref{mass-term}), (\ref{int-term}) ,
\begin{equation}
S\equiv S_0+S_m+S_{int} =S^{\dagger}. 
\end{equation}

We note that 
any two-point functions   with  translational invariance and locality commute with each other in the sense of matrices.  
\begin{eqnarray}
(AB)_{ik} &=& \sum_j A_{ij}B_{jk} = \sum_j A(i-j)B(j-k) \nonumber  \\
&=&\sum_{j'=i+k-j}B(i-j')A(j'-k) 
= \sum_{j'}B_{ij'}A_{j'k}=(BA)_{ik} .
\end{eqnarray}
Indeed, in the  real lattice space, any translationally-invariant and local 
two-point functions including the difference operator  $\nabla$ and the massive propagator $D$
  commute  with each other in the sense of matrices.

\section{On a solution of a linear  $\tilde{\nabla}$  equation 
for coefficients of functionals}
\setcounter{equation}{0}

In  proving a fundamental theorem on cohomology of a nilpotent SUSY, 
we need a general solution of a linear  equation for 
TILCs (translationally-invariant and local coefficients) with $\nabla$. 
In the linear equation of TILCs, a difference operator $\nabla$  and the coefficient can be 
expressed by holomorphic functions with many-variables,
namely holomorphic representation. 
The equation in question is typically  
\begin{equation}
\sum_{a=1}^M \tilde{A}_{a}(z_1,\cdots,z_N)\tilde{\nabla}(z_a)=0,
\label{L-equation}
\end{equation}
\noindent
for $M \leq  N$. 
Note that  
$\nabla$ has two site-indices  and the coefficient has $N+1$ site-indices 
 in lattice site representation. 

The solution of (\ref{L-equation}) can be generally written as
\begin{equation}
\tilde{A}_{a}(z_1,\cdots,z_N)=\sum_{b=1}^{M}\tilde{A}_{\{ab\}}(z_1,\cdots,z_N) 
\tilde{\nabla}(z_b),~\tilde{A}_{\{ab\}}(z_1,\cdots,z_N)= - \tilde{A}_{\{ba\}}(z_1\cdots z_N)
\label{sol}
\end{equation}
\noindent
where  $\tilde{A}_{\{ab\}}$ for $a,b=1,2,\cdots M$ are holomorphic functions in the domain ${\cal D}^N$.

\begin{proof}
We carry out our proof for the case where the $\tilde{\nabla}(z)$ has a finite number of multiple zeros,\footnote{Nonzero $\tilde{\nabla}(z)$ 
has no infinite numbers of zeros on the annulus ${\cal D}$ with a sufficiently small width. The reason is that the closure of 
the annulus domain is compact and the $\tilde{\nabla}(z)$ is a holomorphic function on the domain.} 
with e.g.~a doubling phenomena in mind. 
Let $z_k^{(0)}~(k=1,\cdots,\nu)$ 
denote the zeros of $\tilde{\nabla}(z)$,  where  $\nu$ is the number of the zeros.
We prove the above statement by induction.  
For $M=2$ case, (\ref{L-equation}) becomes 
\begin{equation}
\tilde{A}_1(z_1,\cdots,z_N)\tilde{\nabla}(z_1)+
\tilde{A}_2(z_1,\cdots,z_N)\tilde{\nabla}(z_2)=0
\label{L-equation-2}
\end{equation}
\noindent
where $N \geq 2$. By  reminding $\tilde{\nabla}(z_k^{(0)})=0$ and  
setting $z_2=z_k^{(0)}$ in (\ref{L-equation-2}),  we obtain 
\begin{equation}
\tilde{A}_1(z_1,z_2=z_k^{(0)},z_3,\cdots,z_N)\tilde{\nabla}(z_1)=0. 
\label{zero-1}
\end{equation}
\noindent
Since (\ref{zero-1}) is true on a complex $N$-dimensional domain 
$\{(z_1,\cdots,z_N)
\in \mathcal{D}^{N}
| z_1\ne z_i^{(0)},~i=1,\cdots,\nu \}$,  
the holomorphy of the coefficient $A_1$ leads us to 
\begin{equation}
\tilde{A}_1(z_1,z_2=z_k^{(0)},z_3,\cdots,z_N)=0. 
\label{add-1}
\end{equation}
\noindent
Since (\ref{add-1}) must hold  for every $k$,  
we can write $\tilde{A}_1(z_1,\cdots,z_N)$ into the form
\begin{equation}
\tilde{A}_1(z_1,\cdots,z_N)
=
\frac{\tilde{A}_1(z_1,\cdots,z_N)}{\tilde{\nabla}(z_2)}\,\tilde{\nabla}(z_2)
\equiv
\tilde{B}(z_1,\cdots,z_N)\tilde{\nabla}(z_2)
\label{A_1-2}
\end{equation}
\noindent
where $\tilde{B}$ is a 
holomorphic function in the domain ${\cal D}^N$
because $\tilde{A}_1(z_1,\cdots,z_N)/\tilde{\nabla}(z_2)$ has only 
removable singularities.
From (\ref{L-equation-2}), (\ref{A_1-2}) and 
the holomorphy of $\tilde{A}_2$,   we get  
\begin{equation}
\tilde{A}_2(z_1,\cdots,z_N)=-\tilde{B}(z_1,\cdots,z_N)\tilde{\nabla}(z_1) .
\label{A_2-2}
\end{equation}
\noindent
From (\ref{A_1-2}) and (\ref{A_2-2}), the following equation 
\begin{equation}
\tilde{A}_a(z_1,\cdots,z_N)= \sum_{b=1}^2 \tilde{A}_{\{ab\}}(z_1,\cdots,z_N) \tilde{\nabla}(z_b)
\end{equation}
\noindent
is obtained with $\tilde{A}_{\{ 12 \}}=\tilde{B}=-\tilde{A}_{\{21\}},\tilde{A}_{\{11\}}=\tilde{A}_{\{22\}}=0$ and we find that (\ref{sol}) indeed holds for $M=2$.  

Now let us assume that the above statement is true for $M=m$.  
Namely, for  a linear equation
\begin{equation}
\sum_{a=1}^m \tilde{A}^{(m)}_{a}(z_1,\cdots,z_N)\tilde{\nabla}(z_a)=0 
\qquad (m<N),
\label{L-equation-3}
\end{equation}
\noindent
the general solution is written as 
\begin{equation}
\tilde{A}^{(m)}_{a}(z_1,\cdots,z_N)=\sum_{b=1}^{m}\tilde{A}^{(m)}_{\{ab\}}(z_1,\cdots,z_N) 
\tilde{\nabla}(z_b),~\tilde{A}^{(m)}_{\{ab\}}(z_1,\cdots,z_N)= - \tilde{A}^{(m)}_{\{ba\}}(z_1\cdots z_N) .
\label{sol-m}
\end{equation}
Then we consider  $M=m+1$ case with $m+1<N$, 
\begin{equation}
\sum_{a=1}^{m}\tilde{A}_a^{(m+1)}(z_1,\cdots,z_N)\tilde{\nabla}(z_a)+ \tilde{A}_{m+1}^{(m+1)}(z_1,\cdots,z_N)\tilde{\nabla}(z_{m+1})=0. 
\label{L-equation-4}
\end{equation}
\noindent
From a holomorphy of $\tilde{A}_{m+1}^{(m+1)}$ and 
\begin{equation}
\tilde{A}_{m+1}^{(m+1)}(z_1=z_{k_1}^{(0)},z_2=z_{k_2}^{(0)},\cdots,z_m=z_{k_m}^{(0)},z_{m+1},\cdots,z_N)=0, 
\label{zero-2}
\end{equation}
\noindent
$\tilde{A}_{m+1}^{(m+1)}(z_1,\cdots,z_N)$ is shown to be expressed as 
\begin{equation}
\tilde{A}_{m+1}^{(m+1)}(z_1,\cdots,z_N)=\sum_{a=1}^m \tilde{B}_a^{(m+1)}(z_1,\cdots,z_N)\tilde{\nabla}(z_a) ,
\label{sol-m-1}
\end{equation}
\noindent
where $\tilde{B}_a^{(m+1)} (a=1,\cdots,m)$ are some holomorphic functions.  

As an example, we explain  the above result in the case of $m=2$  below, although  (\ref{sol-m-1}) generally holds.   
For  $m=2$,\ 
 (\ref{zero-2}) becomes 
\begin{equation}
\tilde{A}_{3}^{(3)}(z_1=z_{k_1}^{(0)},z_2=z_{k_2}^{(0)},z_3,\cdots,z_N)= 0. 
\end{equation}
\noindent
Note that the two  zeros $z_{k_1}^{(0)},z_{k_2}^{(0)}$  are not necessarily  the same. 
We substitute this to the following identity 
\begin{eqnarray}
&&\hspace{-8mm}
  \tilde{A}_{3}^{(3)}(z_1,z_2,z_3,\cdots,z_{N}) 
 = \frac{\tilde{A}_{3}^{(3)}(z_1,z_2,z_3,\cdots,z_{N})
    -{\textstyle\sum\limits_{k_{1}=1}^{\nu}}
    f_{k_{1}}(z_1)\tilde{A}_{3}^{(3)}(z_{k_1}^{(0)},z_2,z_3,\cdots,z_{N})}{
    \tilde{\nabla}(z_1)} \tilde{\nabla}(z_1) \nonumber \\
&&\hspace{-4mm}
  + \sum_{k_1=1}^{\nu}\frac{f_{k_1}(z_1)
  \tilde{A}_{3}^{(3)}(z_{k_1}^{(0)},z_2,z_{3},\cdots,z_{N})
   -{\textstyle\sum\limits_{k_{2}=1}^{\nu}}
   f_{k_1}(z_1)f_{k_2}(z_2)
   \tilde{A}_{3}^{(3)}(z_{k_1}^{(0)},z_{k_2}^{(0)},z_3,\cdots,z_{N})}{
   \tilde{\nabla}(z_2)}\tilde{\nabla}(z_2) \nonumber \\
&&\hspace{-4mm}
   + \sum_{k_1,k_2=1}^{\nu}f_{k_1}(z_1)f_{k_2}(z_2)
   \tilde{A}_{3}^{(3)}(z_{k_1}^{(0)},z_{k_2}^{(0)},z_3,\cdots,z_{N}) .
\end{eqnarray}
\noindent
Here the Lagrange's interpolation functions
\begin{equation}
f_k(z)\equiv \frac{\prod_{j\ne k}(z-z_{j}^{(0)})}{\prod_{j\ne k} (z_{k}^{(0)}- z_{j}^{(0)}) },~f_k(z_{j}^{(0)})=\delta_{kj}  
\end{equation}
\noindent
are used. The result is 
\begin{eqnarray}
&&\hspace{-8mm}
  \tilde{A}_{3}^{(3)}(z_1,z_2,z_3,\cdots,z_{N}) 
 = \frac{\tilde{A}_{3}^{(3)}(z_1,z_2,z_3,\cdots,z_{N})
    -{\textstyle\sum\limits_{k_{1}=1}^{\nu}}
   f_{k_1}(z_1) \tilde{A}_{3}^{(3)}(z_{k_1}^{(0)},z_2,z_3,\cdots,z_{N})}{
    \tilde{\nabla}(z_1)} \tilde{\nabla}(z_1) \nonumber \\
&&\hspace{-8mm}
  + \sum_{k_1=1}^{\nu}\frac{f_{k_1}(z_1)
  \tilde{A}_{3}^{(3)}(z_{k_1}^{(0)},z_2,z_{3},\cdots,z_{N})
   -{\textstyle\sum\limits_{k_{2}=1}^{\nu}}
   f_{k_1}(z_1)f_{k_2}(z_2)
   \tilde{A}_{3}^{(3)}(z_{k_1}^{(0)},z_{k_2}^{(0)},z_3,\cdots,z_{N})}{
   \tilde{\nabla}(z_2)}\tilde{\nabla}(z_2) .\nonumber\\
   \label{nabla-1}
\end{eqnarray}
\noindent
Since  $z_1=z_{k_1}^{(0)}$ and $z_2=z_{k_2}^{(0)}$ are removable singularities on the right hand side of (\ref{nabla-1}) and 
the multiplication  of $f_k(z)$ for holomorphic functions does not change the holomorphy on the annulus,  
we obtain 
\begin{equation}
\tilde{A}_{3}^{(3)}(z_1,z_2,z_3,\cdots,z_N)=\sum_{a=1}^2 \tilde{B}_a^{(3)}(z_1,z_2,z_3,\cdots,z_N)
\tilde{\nabla}(z_a) ,
\label{sol-m-2}
\end{equation}
\noindent
where $\tilde{B}_a^{(3)}$ are holomorphic functions in ${\cal D}^3$. 

From (\ref{sol-m-1}), (\ref{L-equation-4}) can be regarded as an inhomogeneous linear equation for 
$\tilde{A}_a^{(m+1)}$. Thus, the  general solution for the equation can be expressed as 
\begin{eqnarray}
&&\hspace{-10mm} 
  \tilde{A}_{a}^{(m+1)}(z_1,\cdots,z_N)
   =-\tilde{B}_a^{(m+1)}(z_1,\cdots,z_N)\tilde{\nabla}(z_{m+1})
    + \sum_{b=1}^m\tilde{A}_{\{ab\}}^{(m+1)}(z_1,\cdots,z_N) \tilde{\nabla}(z_b) ,
    \nonumber\\
&&\hspace{90mm} (a=1,2,\cdots,m)  
\end{eqnarray}
\noindent
where 
the 1st term on the right-hand-side gives a special solution to 
(\ref{L-equation-4}) and the 2nd term corresponds to the general solution
for a homogeneous equation
${\textstyle\sum\limits_{a=1}^{m}}\tilde{A}_a^{(m+1)}(z_1,\cdots,z_N)
\tilde{\nabla}(z_a)=0$
(see (\ref{sol-m})).
%
By defining 
\begin{eqnarray}
%
%
\tilde{A}_{\{m+1 a\}}^{(m+1)}(z_1,\cdots,z_N)& \equiv &\tilde{B}_{a}^{(m+1)}(z_1,\cdots,z_N) \nonumber \\
\tilde{A}_{\{a m+1\}}^{(m+1)}(z_1,\cdots,z_N) & \equiv & -\tilde{B}_{a}^{(m+1)}(z_1,\cdots,z_N) \nonumber \\
\tilde{A}_{\{m+1 m+1\}}^{(m+1)} (z_1,\cdots,z_N) & \equiv &  0 ,
\end{eqnarray}
\noindent
we obtain the general solution (\ref{sol}) for $M=m+1$.

\end{proof}

By repeating the above argument, the following corollary is immediately obtained. 
If 
\begin{equation}
\sum_{a_1=1}^M\tilde{A}_{\{a_1 \cdots  a_{n}\}}(z_1,\cdots,z_N)\tilde{\nabla}(z_{a_{1}})=0, 
\end{equation}
\noindent
for $M\leq N$, 
then 
\begin{equation}
\tilde{A}_{\{a_1 \cdots a_{n}\}}(z_1,\cdots,z_N)
=\sum_{a_{n+1}=1}^M\tilde{A}_{\{a_1 \cdots a_{n+1} \}}(z_1,\cdots,z_N)\tilde{\nabla}(z_{a_{n+1}}).
\end{equation}
Here all $\tilde{A}$ are holomorphic functions in a domain ${\cal D}^N$.  

In applying these results to the proof of the fundamental theorem of $Q_-$-cohomology, we need the lattice site representation 
in which a linear equation for TILC 
\begin{equation}
\sum_{a=1}^M\sum_{m} B_{k,m,n_1\cdots \hat{n}_a \cdots n_N}\nabla_{m,n_a} =0
\end{equation}
\noindent 
has a general  translationally invariant and local  solution 
\begin{equation}
B_{k,m,n_1\cdots \hat{n}_a \cdots n_N}=\sum_{b=1,b\ne a}^M\sum_{m'} C_{k,\{mm'\},n_1\cdots \hat{n}_a \cdots \hat{n}_b
\cdots n_N}\nabla_{m',n_b} ,
\label{sol-site}
\end{equation}
\noindent
where $C$'s are some TILC.
Although the index $k$ seems redundant, at least one extra index is necessary for writing down corresponding lattice site representation, since the number of independent indices are one less than the total number of indices for the translationally invariant coefficient.
%


%
\section{Proof of fundamental theorem on the cohomology of nilpotent SUSY}
\setcounter{equation}{0}

Since the analysis of $-$type fields is similar to that of  $+$type fields,   we concentrate on only $+$ fields ($+$type functionals) and $Q_-$ in this appendix. 
Thus,  we omit the subscript of fields  as $\phi(=\phi_+), \chi(=\chi_+), F(=F_+),\bar{\chi}(=\bar{\chi}_+)$. 
We set the $U(1)_R$ charge $R$ of  ${\cal O}$ as  $R=n_+-1-2K$ where $n_+$ and $K$ are the number of $+$type fields and  
the combined number of $F$ and  $\bar{\chi}$, respectively.  
When $K=0$,  namely $R=R_{\rm max}\equiv n_+-1$, 
the functional  ${\cal O} = \sum C \chi \phi\cdots \phi$ is called as a CLR term with  $N_F=1$. 
The $N$ is defined as 
the total number of $\phi$ and  ${\chi}$.  $K$ and $N$ are conserved  under $Q_-$ transformation (\ref{SUSYtrans-2}).


\textit{Fundamental theorem  of  $Q_{-}$-cohomology} says 
if    $Q_-{\cal O}=0$ for a $+$type local functional  ${\cal O}$ in a sector of $N_F=1$,  in the case of  $R < R_{\rm max}$, 
${\cal O}=Q_{-}{\cal P}$ with a local functional ${\cal P}$, and in the case of $R=R_{\rm max}$,     
${\cal O}$ can be a CLR term with $N_F=1$ up to some $Q_-$-exact local functionals.  


This theorem leads us that  CLR terms ($R=R_{\rm max}$)  are only  nontrivial cohomology candidates 
 in the fermion number 1 sector.  Namely, it is a fundamental theorem 
 on cohomology of nilpotent SUSY. 
This theorem can be proved for both original fields and  superfields although we carry out 
for original fields here.


\begin{proof}

Since $K,N$ are conserved numbers under $Q_-$-transformation,  it is sufficient to   consider 
the following functionals 
as  fermion number +1 translationally invariant, local and general functionals%
\begin{eqnarray}
{\cal O}(K,N)=\sum_{\bm{k},\bm{\ell},\bm{m},\bm{n}}\sum_{p=0}^{\min(K,N)} 
B^{(p)}_{(k_1\cdots k_{K-p})\{\ell_1\cdots\ell_p\}\{m_1\cdots m_{p+1}\}
(n_1\cdots n_{N-p})} \nonumber \\
\times F_{k_1}\cdots  F_{k_{K-p}}
 \bar{\chi}_{\ell_1}\cdots \bar{\chi}_{\ell_{p}}  \chi_{m_1}\cdots \chi_{m_{p+1}}  \phi_{n_1}\cdots \phi_{n_{N-p}}  .
 \label{F=1 functional}
\end{eqnarray}
\noindent
where $B^{(p)}$ is a TILC and bold indices $\bm{k},\bm{\ell},\bm{m},\bm{n}$ stand for multi indices
\begin{eqnarray}
\bm{k}&\equiv& k_1,\cdots,k_{K-p},\quad
\bm{\ell}\equiv \ell_1,\cdots,\ell_{p}, \nonumber \\
\bm{m} &\equiv&  m_1,\cdots,m_{p+1},\quad
\bm{n} \equiv  n_1,\cdots,n_{N-p}.
\end{eqnarray}

 $B^{(\min{K,N})+1}\equiv 0$ is set for convenience,   
We impose the $Q_-$-invariant (closed form) condition 
\begin{equation}
Q_- {\cal O}(K,N)=0
\label{closed condition}
\end{equation}
\noindent
on ${\cal O}$.  Using a single dot-symbol for absence of the corresponding  index such as $\{ \cdot\}$,   
the condition (\ref{closed condition})  is translated into the condition on the TILCs,
\begin{equation}
\sum_{j=1}^{N+1} \sum_{m}B^{(0)}_{(k_1\cdots k_{K})\{\cdot\}\{m\}
(n_1\cdots  \hat{n}_j \cdots  n_{N+1})}  \nabla_{m n_j}   =0 ,
 \label{closed condition 1}   
\end{equation}
for $p=0$ and
\begin{eqnarray}
 \frac{p+1}{N-p+1}
 \sum_{j=1}^{N-p+1} \sum_{m}
B^{(p)}_{(k_1\cdots k_{K-p})\{\ell_1\cdots\ell_p\}\{m_1\cdots m_p m\}(n_1\cdots  \hat{n}_j \cdots n_{N-p+1})}
\nabla_{m n_j} \nonumber \\
+  \frac{K-p+1}{p}
\sum_{j=1}^{p} (-1)^j \sum_{k}
B^{(p-1)}_{(k_1\cdots k_{K-p}k)\{\ell_1\cdots \hat{\ell}_j \cdots  \ell_p\}\{m_1\cdots m_p\}
(n_1\cdots n_{N-p+1})}\nabla_{k\ell_j}=0  ,
\label{closed condition 2}
\end{eqnarray}
for  $1\le p \le \min(K,N)$,
\noindent
where we used a hat ($~\hat{}~$) symbol for absent index. 
For $p=\min(K,N)+1=N+1$ ($K>N$case), (\ref{closed condition 2}) becomes
\begin{equation}
\sum_{j=1}^{N+1} \sum_{k}(-1)^j
B^{(N)}_{(k_1\cdots k_{K-N-1}k)\{\ell_1\cdots  \hat{\ell}_j \cdots\ell_{N+1}\}\{m_1\cdots m_{N+1}\}(\cdot)} \nabla_{k\ell_j} =0.
\label{closed condition 3}
\end{equation}
\noindent
In the case of  $p=\min(K,N)+1=K+1$ ($K\le N$ case),   there is no extra condition  such as (\ref{closed condition 3})
for  $B^{(K)}$, since  $B^{(K+1)}=0$ and  $K-p+1=0$  in (\ref{closed condition 2}).  
Note that  (\ref{closed condition 2}) can be solved as inhomogeneous linear equations for $B^{(p)}$ 
except for $K=0$.  

Then we discuss solutions of the above conditions for $K=0$ and $K\ge 1$ cases separately:
\begin{enumerate}
 \item  $K \ge 1$ case 
 
To express a general solution, we introduce  the following TILCs
 \begin{equation}
C^{\{p\}}_{(k_1\cdots k_{K-p+1})\{\ell_1\cdots \ell_{p-1}\}\{m_1\cdots m_{p+1}\}(n_1\cdots n_{N-p})} ,
\end{equation}
\noindent
and
\begin{equation}
C^{\{0\}}=0 .
\end{equation}
\noindent
From the symmetric property of $C^{\{p\}}$, the coefficient satisfies  the following properties, 

\begin{equation}
\sum_{m,n}\sum_{i,j=1,i\ne j}^{N-p+2}
C^{\{p\}}_{(k_1\cdots k_{K-p+1})\{\ell_1\cdots \ell_{p-1}\}\{m_1\cdots m_{p-1}mn\}(n_1\cdots \hat{n}_i \cdots  \hat{n}_{j} \cdots n_{N-p+2})} 
\nabla_{mn_i}  \nabla_{n n_j}  = 0 ,
\label{prop-1}
\end{equation}
\begin{equation}
\sum_{k,\ell}\sum_{i,j=1,i\ne j}^{p+1}(-1)^{i+j}
C^{\{p\}} _{(k_1\cdots k_{K-p-1}k\ell)\{\ell_1\cdots \hat{\ell}_i \cdots \hat{\ell}_j \cdots \ell_{p+1}\}\{m_1\cdots m_{p+1}\}(n_1 \cdots n_{N-p})} \nabla_{k\ell_i} 
 \nabla_{\ell \ell_{j}}  = 0 ,
\label{prop-2}
\end{equation}
\begin{eqnarray}
&&\sum_{k,m}\sum_{i=1}^{N-p}\sum_{j=1}^{p+1} (-1)^j
C^{\{p+1\}} _{(k_1\cdots k_{K-p-1}k)\{\ell_1\cdots  \hat{\ell}_j \cdots \ell_{p+1}\}\{m_1\cdots m_{p+1}m\}(n_1 \cdots  \hat{n}_i \cdots n_{N-p})} 
\nabla_{k\ell_j}  \nabla_{mn_i}  \nonumber \\
&=&\sum_{k,m}\sum_{i=1}^{p+1}\sum_{j=1}^{N-p}(-1)^i
C^{\{p+1\}}_{(k_1\cdots k_{K-p-1}k)\{\ell_1\cdots \hat{\ell}_i  \cdots \ell_{p+1}\}\{m_1\cdots m_{p+1}m\}(n_1 \cdots  \hat{n}_j \cdots n_{N-p})}  
 \nabla_{mn_j}  \nabla_{k\ell_i}.  \nonumber \\
\label{prop-3}
\end{eqnarray}
From  a solution  (\ref{sol-site}),  
the general solution of TILC
$B^{(0)}$ for (\ref{closed condition 1})
is written as 
\begin{equation}
B^{(0)}_{(k_1\cdots k_{K}) \{\cdot\}\{m\}
(n_1\cdots  \hat{n}_j \cdots n_{N+1})}=
\sum_{i=1,i\ne j}^{N+1}\sum_{m'} C^{\{1\}}_{(k_1\cdots k_{K})\{\cdot\} \{ m m' \}
(n_1\cdots \hat{n}_i \cdots \hat{n}_j \cdots n_{N+1})}
\nabla_{m' n_i} .
\label{general-B}
\end{equation}
\noindent 
A general solution for $B^{(1)}$ and $B^{(2)}$ are  
\begin{eqnarray}
&&B^{(1)}_{(k_1\cdots k_{K-1})\{\ell\}\{m_1m_{2}\}
(n_1\cdots n_{N-1})} \nonumber \\
&=&\frac{KN}{1\cdot2}\sum_{k'} C^{\{1\}}_{(k_1\cdots k_{K-1}k') \{\cdot\}\{ m_1 m_2 \}
(n_1\cdots  n_{N-1})}
\nabla_{k'\ell}  \nonumber \\
&+& \sum_{m'}\sum_{j=1}^{N-1} C^{\{2\}}_{(k_1\cdots k_{K-1})\{\ell\}\{m_1m_2m'\}(n_1\cdots \hat{n}_j  \cdots n_{N-1})} \nabla_{m'n_j}
\end{eqnarray}
\noindent
and
\begin{eqnarray}
&&B^{(2)}_{(k_1\cdots k_{K-2})\{\ell_1\ell_2\}\{m_1m_2 m_{3}\}
(n_1\cdots n_{N-2})} \nonumber \\
&=&-\frac{(K-1)(N-1)}{2\cdot3} \sum_{k'}\sum_{i=1}^2  (-1)^i
C^{\{2\}}_{(k_1\cdots k_{K-2}k')\{\ell(\hat{\ell}_i)\}\{m_1m_2m_3\}(n_1\cdots  n_{N-2})}\nabla_{k'\ell_i} \nonumber \\
&+& \sum_{m'}\sum_{j=1}^{N-2}C^{\{3\}} _{(k_1\cdots k_{K-2})\{\ell_1\ell_2\}\{m_1m_2m_3m'\}(n_1\cdots \hat{n}_j \cdots n_{N-2})} \nabla_{m'n_j},
\end{eqnarray}
\noindent
where we used    (\ref{sol-site}),  (\ref{prop-1}), (\ref{prop-2}), and (\ref{prop-3}) and $\ell(\hat{\ell}_i)$ means  $\ell_2$ for $i=1$ and  $\ell_1$ for $i=2$.
Consequently, we get a general solution for $B^{(p)}$
\begin{eqnarray}
&&B^{(p)}_{(k_1\cdots k_{K-p})\{\ell_1\cdots\ell_p\}\{m_1\cdots m_{p+1}\}
(n_1\cdots n_{N-p})} \nonumber \\
&=&-\frac{(K-p+1)(N-p+1)}{p(p+1)} \sum_{k'} \sum_{i=1}^{p}  (-1)^i
C^{\{p\}}_{(k_1\cdots k_{K-p}k')\{\ell_1\cdots  \hat{\ell}_i  \cdots \ell_{p}\}\{m_1\cdots m_{p+1}\}(n_1\cdots n_{N-p})}    \nabla_{k'\ell_i} \nonumber \\
&+& \sum_{m'} \sum_{j=1}^{N-p}  
C^{\{p+1\}} _{(k_1\cdots k_{K-p})\{\ell_1\cdots \ell_{p}\}\{m_1\cdots m_{p+1}m'\}(n_1\cdots  \hat{n}_j \cdots n_{N-p})}  \nabla_{m'n_j} .
\label{p-solution}
\end{eqnarray}
\noindent
There remains an extra condition (\ref{closed condition 3}).  
For $K > N$ and $p=N$,   the condition (\ref{closed condition 3}) reads 
\begin{eqnarray}
&&B^{(N)}_{(k_1\cdots k_{K-N})\{\ell_1\cdots\ell_N\}\{m_1\cdots m_{N+1}\}(\cdot)} \nonumber \\
&=&-\frac{(K-N+1)}{N(N+1)}\sum_{k'} \sum_{i=1}^{N}  (-1)^i
C^{\{N\}}_{(k_1\cdots k_{K-N}k')\{\ell_1\cdots  \hat{\ell}_i  \cdots \ell_{N}\}\{m_1\cdots m_{N+1}\}(\cdot)}    \nabla_{k'\ell_i}.  \nonumber \\
&&
\label{N-solution}
\end{eqnarray}
\noindent
By directly solving (\ref{closed condition 1}), we also have
\begin{eqnarray}
B^{(0)}=0 
\label{B0}
\end{eqnarray}
\noindent
for $N=0$. 
It is consistent with $N=0$ case of (\ref{N-solution}).  

Finally, from (\ref{F=1 functional}), (\ref{p-solution}) and  (\ref{N-solution}),  we can generally obtain an exact form
\begin{eqnarray}
{\cal O}(K,N)&=&Q_-\Big( \sum_{\bm{k},\bm{\ell},\bm{m}',\bm{n}'} \sum_{p=0}^{\min(K,N-1)}  
 \frac{i(N-p)}{(p+2)}  C^{ \{ p+1\}}_{(\bm{k})\{\bm{\ell}\}
\{\bm{m}'\}(\bm{n}')}  \nonumber \\
&&\times F_{k_1}\cdots F_{k_{K-p}}
 \bar{\chi}_{\ell_1}\cdots  \bar{\chi}_{\ell_p}
  {\chi}_{m_1}\cdots  {\chi}_{m_{p+2}}
 \phi_{n_1}\cdots \phi_{n_{N-p-1}}
\Big).
\end{eqnarray}
\noindent
where $\bm{m}'\equiv m_1,\cdots, m_{p+2},~\bm{n}'\equiv n_1,\cdots, n_{N-p-1}$. 
In $N=0$ case,  there is no $Q_-$-invariant translationally invariant and local functional, which is consistent with (\ref{B0}).

 \item  $K=0$ case

In  $K=0$, the functional can be written as 
\begin{equation}
{\cal O}(0,N)=\sum_{m_1,n_1,\cdots,n_N}B^{(0)}_{(\cdot)\{\cdot\}\{m_1\}(n_1\cdots n_{N})}   \chi_{m_1}\phi_{n_1}\cdots \phi_{n_N} 
\label{CLR-apndx}
\end{equation}
\noindent
and the $Q_-$-invariant condition corresponds to a cyclic Leibniz rule(CLR) 
\begin{equation}
\sum_{a=1}^{N+1}\sum_{n_1,\cdots,n_{N+1} }B^{(0)}_{(\cdot)\{\cdot\}\{n_{a}\}(n_1\cdots n_{a-1}n_{a+1}\cdots n_{N+1})}  
\phi_{n_1}\cdots \phi_{n_{a-1}}   (\nabla\phi)_{n_{a}}  \phi_{n_{a+1}}\cdots \phi_{n_{N+1}} =0. 
\label{CLR-apndx-2}
\end{equation}
\noindent
This term (\ref{CLR-apndx}) is just what is called {\it CLR term with $N_F=1$} in section 3. 
For the coefficient $B^{(0)}$ to be proportional to $\nabla$, at least one irrelevant index is necessary, just  as the last comment of appendix B.  
But all $N+1$ indices of $B^{(0)}$ are relevant for the contraction with index of $\nabla$, thereby the previous argument does not apply here.
So it implies that $B^{(0)}$ in (\ref{CLR-apndx})
is not always  written like (\ref{general-B}), i.e.~exact form in $K=0$.   
\end{enumerate}
\end{proof}



\end{document}